\documentclass[conference]{IEEEtran}
\IEEEoverridecommandlockouts
\pdfoutput=1 
\usepackage{cite}
\usepackage{amsmath,amssymb,amsfonts,amsthm,nccmath,amstext}
\usepackage{algorithm}
\usepackage{graphicx}
\usepackage{textcomp}
\usepackage{ulem}
\usepackage{xcolor}
\def\BibTeX{{\rm B\kern-.05em{\sc i\kern-.025em b}\kern-.08em
    T\kern-.1667em\lower.7ex\hbox{E}\kern-.125emX}}
    \usepackage{multirow}

\usepackage{threeparttable}
\usepackage{algorithmicx}
\usepackage{algpseudocode}
\usepackage{color}
\usepackage{pifont}
\usepackage{subcaption}
\usepackage{hyperref}

\newcommand{\xmark}{\ding{55}}
\usepackage{microtype}
\usepackage{wallpaper}
\usepackage{lipsum}
\theoremstyle{definition}
\newtheorem{definition}{Definition}[]
\newtheorem{remark}{Remark}[]
\newtheorem{theorem}{Theorem}[]
\newtheorem{corollary}{Corollary}[theorem]
\newtheorem{lemma}[theorem]{Lemma}
\newenvironment{sproof}{%
  \proof}{\endproof}

\begin{document}

\title{SPOT: Secure and Privacy-preserving prOximiTy protocol for e-healthcare systems}

\author{\IEEEauthorblockN{Souha Masmoudi}
\IEEEauthorblockA{
\textit{Samovar, Telecom SudParis,} \\
\textit{Institut Polytechnique de Paris} \\
Palaiseau, France \\
\textit{Member of the Chair Values and} \\
\textit{Policies of Personal Information} \\
ORCID: 0000-0002-7194-8240}
\and
\IEEEauthorblockN{Nesrine Kaaniche}
\IEEEauthorblockA{
\textit{Samovar, Telecom SudParis,} \\
\textit{Institut Polytechnique de Paris} \\
Palaiseau, France \\
\textit{Member of the Chair Values and} \\
\textit{Policies of Personal Information} \\
ORCID: 0000-0002-1045-6445}
\and
\IEEEauthorblockN{Maryline Laurent}
\IEEEauthorblockA{
\textit{Samovar, Telecom SudParis,} \\
\textit{Institut Polytechnique de Paris} \\
Palaiseau, France \\
\textit{Member of the Chair Values and} \\
\textit{Policies of Personal Information} \\
ORCID: 0000-0002-7256-3721}

}

\maketitle

\begin{abstract}
This paper introduces \textsf{SPOT}, a Secure and Privacy-preserving prOximity based protocol for e-healthcare systems. It relies on a distributed proxy-based approach to preserve users' privacy and a semi-trusted computing server to ensure data consistency and integrity. 
The proposed protocol ensures a balance between security, privacy and scalability.
As far as we know, in terms of security, \textsf{SPOT} is the first one to prevent malicious users from colluding and generating false positives. In terms of privacy, \textsf{SPOT} supports both anonymity of users being in proximity of infected people and unlinkability of contact information issued by the same user. 
A concrete construction based on structure-preserving signatures and NIWI proofs is proposed and a detailed security and privacy analysis proves that \textsf{SPOT} is secure under standard assumptions.
In terms of scalability, \textsf{SPOT}'s procedures and algorithms are implemented to show its efficiency and practical usability with acceptable computation and communication overhead.
\end{abstract}


\section{Introduction}
\label{sec:introduction}
The recent health crisis has led governments to apply different tracing solutions to control the contamination chain among the population. 
These solutions are aimed at sharing valuable data while preserving users' privacy. However, there are still privacy threats and robustness challenges as long as users are required to disclose and share correct sensitive and personal information with different third parties with various levels of trust. 

Most of the solutions rely on the Bluetooth technology, namely Bluetooth Low Energy (BLE), to exchange contact information, thanks to its efficiency in active communications \cite{reichert2020survey}.
Among these governmental solutions, the \textsf{TraceTogether} application \cite{TraceTogether} has been launched by Singapore. \textsf{TraceTogether} enables to collect, via the Bluetooth technology, temporary IDs (generated by a central trusted server) of users in close proximity. Collected IDs are stored in an encrypted form using the server’s public key, at users' devices, and in case of infection, they are shared with the server.
The \textsf{COVIDSafe} application \cite{COVIDSafe} from the Australian government is another Bluetooth-based solution. It also logs encrypted users' contact information, and share them once an infection is detected. The server is required to alert users at risk without revealing the identity of the infected user. Both \textsf{TraceTogether} and \textsf{COVIDSafe} applications are set upon a centralized architecture.
Many other applications like \textsf{Stop COVID-19} (Croatia) \cite{stopcovid19}, \textsf{CA Notify} (California) \cite{canotify} rely on the \textsf{Google and Apple Exposure Notification (GAEN)} service \cite{applegoogle}, which is set upon a decentralized architecture. Although contact tracing applications have helped governments to alleviate the widespread of the pandemic by automating  the manual contact tracing done by health authorities, they raise critical privacy concerns, namely users tracking and identification \cite{leith2021contact}.

Academic solutions have been also proposed to support both centralized \cite{robert} and decentralized \cite{dp3t,Chan2020PACTPP,pact2,covid2020zk,pietrzak2020} architectures. However, each architecture has his merits and limits in terms of security and privacy. Indeed, using centralized solutions, users guarantee the reception of correct alerts as long as the generation of users' contact tokens and the verification in case of infection are performed by a centralized server. This guarantee is compensated with threats to users' privacy, i.e., users are exposed to tracking and identification of their contact lists. Decentralized solutions have been proposed to mitigate these privacy threats. Users are responsible for generating their contact tokens in order to ensure their privacy and anonymity, but, they are not prevented from forging contact information, which results in high level of false alerts. To get the best of both architectures, hybrid architecture based solutions \cite{desire,hoepman2021hansel} have been proposed. However, security and privacy requirements, like the correctness of contact information and the anonymity of contacted users, are not yet handled and ensured together.

In this paper, we present \textsf{SPOT}, a novel hybrid Secure and Privacy-preserving prOximity-based protocol for e-healthcare systems. It combines a decentralized proxy front-end architecture, ensuring both users' anonymity and contact information integrity, and a centralized back-end computing server, guaranteeing a real time verification of contact information integrity. \textsf{SPOT} assumes that two users in close proximity rely on their Ephemeral Bluetooth IDentifiers (EBID) to compute a common contact message. 
This message is relayed to a central server through a group of proxies. 
With the help of the computing server and relying on a proof-based group signature, \textsf{SPOT} prevents users from forging their contact lists. 
The signed contact messages are given to the user to be locally stored. 
In case of a detected infection, the user consents to share his contact list, i.e., a set of signed contact messages, with the health authority. This latter checks the correctness of the received list and shared it back with all the involved users, if the verification holds. The contributions of this paper are summarized as follows:

\begin{itemize}
    \item we design a proximity protocol for e-health services that prevents the injection of false positives, i.e., alert users to be at risk when they are not. \textsf{SPOT} enforces the verification of the correctness and the integrity of users' contact information by health authorities, thanks to the support of both a computing server and a group of proxies.
    \item we guarantee strong privacy properties namely the anonymity of users being in contact with infected people, and the unlinkability of users' transactions when relying on random EBIDs that can neither be linked to each other nor be linked to their issuer.
    \item we propose a concrete construction of the \textsf{SPOT} protocol relying on a structure-preserving signature scheme \cite{XCGSIG} that supports securely signing group's elements, i.e., contacts' information. This signature scheme is coupled with Groth-Sahai Non-Interactive Witness-Indistinguishable (NIWI) proof \cite{grothsahai} in order to ensure integrity of proxies' keys. Indeed, NIWI proofs guarantee the anonymity of proxies while the health authority is still able to verify that proxies are trustful by verifying the validity of their keys without having access to them.
    \item we evaluate the performances of \textsf{SPOT} through the full implementation of different procedures and algorithms. The conducted experiments have shown acceptable computation times proving the efficiency and practical usability of the proposed solution.
\end{itemize}

This paper is organized as follows. 
Section \ref{sec:relatedwork} introduces and compares most closely-related proximity tracing algorithms and solutions to \textsf{SPOT}. Section \ref{sec:spec} gives an overview of \textsf{SPOT}. After introducing the underlying building blocks in Section \ref{sec:blocks}
The concrete construction of the proposed \textsf{SPOT} protocol is presented in Section \ref{sec:constr}. Section \ref{sec:secanal} and Section \ref{sec:perf} provide security and privacy properties and a detailed discussion of \textsf{SPOT}'s conducted experiments, before concluding in Section \ref{sec:conclusion}.

\section{Related Work}
\label{sec:relatedwork}

Recently, several industrial and research contact tracing solutions have been proposed for e-health applications \cite{surveycovid,reichert2020survey}. These solutions aim at ensuring security properties, namely \textbf{anti-replay} mitigating the multi-submission of the same contact information, and \textbf{unforgeability} preventing malicious entities\footnote{Malicious entities involve either a single malicious adversary or colluding adversaries.} from threatening data integrity. Privacy properties have been significantly addressed, including the \textbf{anonymity} of end-users and the \textbf{unlinkability} between their different transactions (i.e., a formal definition of security and privacy requirements is given in Section \ref{sec:thrmodel}).

Indeed, researchers from Inria, France, and Fraunhofer, Germany proposed \textsf{Robert} \cite{robert}, a contact tracing protocol that relies on a centralized architecture, where a central server delivers pseudonyms to users. Each user collects pseudonyms of users in close proximity and shares them with the server when being infected. In such centralized solution, users are sure that they receive correct alerts (i.e., collected pseudonyms are neither replayed nor falsified by malicious entities), however, their privacy is threatened as long as the server is able to identify users' contacts and to track them.
In \cite{dp3t}, Troncoso \emph{et al.} introduced the Decentralized Privacy-Preserving Proximity Tracing (\textsf{DP-3T}) solution which is one of the most popular contact-tracing protocols. \textsf{DP-3T} has been proved to mitigate the privacy threats of centralized solutions as there is no need for a central entity which collects users' contact information with the risk of tracking them. However, it does not prevent relay and replay attacks and gives no mean to verify the correctness of contact information. Thus, users are exposed to false alerts from malicious entities either by creating falsified information or replaying information of previous sessions.
Afterwards, Castelluccia \emph{et al.} proposed \textsf{Desire} \cite{desire}, a proximity tracing protocol that leverages the advantages of the centralized and decentralized solutions. However, some security and privacy issues have not been considered in this solution. First malicious users are able to collude and merge their contact lists, which leads to false positive injection attacks. Second, the server requested to compute the exposure status and risk, is able to link users' requests, and to de-anonymize them. 
Two very similar proposals called \textsf{PACT}s are also introduced. The east coast \textsf{PACT} \cite{pact2} and the west coast \textsf{PACT} \cite{Chan2020PACTPP} are very close to \textsf{DP-3T}. The two solutions rely on random pseudonyms derived from a private seed, that are broadcasted to users in proximity via Bluetooth. The pseudonyms are generated using cryptographic pseudorandom generators and pseudorandom functions. Apart from the non-resistance against replay attacks, these two proposals give no mean to check the correctness of the contact information before being broadcasted. Pietrzak \cite{pietrzak2020} proposed a decentralized contact-tracing solution to mitigate replay attacks against \textsf{DP-3T}. However, privacy concerns are raised, namely tracking users, as geo-location and time of contacts are shared within the Bluetooth message. In \cite{hoepman2021hansel}, Hoepman proposed two tracing protocols, the first one relies on an interactive session between two users in proximity to register contact information. If the interaction fails, the contact is not registered. Thus, the second protocol comes to mitigate this risk of failure and relies on an authority that relays information between users. In both protocols, the identities of users who have been in contact with an infected person, are revealed to a central entity, which contradicts the defined anonymity requirement. 
Liu et al. \cite{covid2020zk} use zero-knowledge proofs and group signatures in order to preserve users' privacy for their proposed tracing protocol. Zero-knowledge proofs are generated by users over the contact information they collected. Indeed, users prove the contacts to their doctors without revealing the information.
Afterwards, doctors, being members of a group, sign the proofs and publish them in a public board. Then, relying on their secrets, users can check if they were in contact with infected people. As such, no entity can identify the contacts of an infected user. However, based on a long interactive protocol between two devices, the collection of contacts' information may result in a failed interaction, thus causing the non-registration of the contact. Furthermore, the authors only consider honest but curious adversaries, which leads to possible false alerts due to malicious colluding users. 
It is also worth noticing that, unlike \textsf{SPOT} and other related work \cite{desire,dp3t,pietrzak2020}, \cite{covid2020zk} assume that all the computations (handshake, zero-knowledge, verification) are performed by the user's device, which leads to device's battery depletion. A new contact tracing strategy based on online social networking is proposed in \cite{TraceMe2022} but does not provide privacy guarantees.
Table \ref{tab:comp} provides a comparative summary between \textsf{SPOT} and related works in terms of architecture settings and properties. As shown, \textsf{SPOT} is the only solution which supports strong security and privacy requirements.

\begin{table*}[!ht]
\begin{center}
\caption{Comparison between \textsf{SPOT} and related works}

\begin{threeparttable}
\small
\begin{tabular}{c|c|c|c|c|c|c|c|c|c}
\hline
     & &  \textsf{SPOT} & \cite{covid2020zk} & \cite{hoepman2021hansel} & \cite{dp3t} & \cite{pact2} \cite{Chan2020PACTPP} & \cite{robert}  & \cite{desire} & \cite{pietrzak2020} \\ \hline
      \multirow{3}{*}{Architecture}  & Centralized & - & - & - & - & - & \checkmark & - & - \\ \cline{2-10}
    & Decentralized & - & \checkmark & - &  \checkmark & \checkmark & - & - & \checkmark \\ \cline{2-10} 
    & Semi-centralized & \checkmark & - & \checkmark &   - & - & -  & \checkmark  &   - \\  \hline
    \multirow{4}{*}{Properties} & Unforgeability  \tnote{a} & \checkmark & \checkmark \tnote{b} & \checkmark \tnote{b} & \xmark & \xmark & N.A. & \xmark & \xmark  \\ \cline{2-10}
    & Anti-replay & \checkmark & \checkmark & \checkmark  & \xmark  & \xmark & \checkmark & \checkmark & \checkmark \\ \cline{2-10}
    & Unlinkability & \checkmark &  N.A. & \xmark  & \checkmark & \checkmark  & \xmark & \checkmark & \xmark \\  \cline{2-10} 
    & Anonymity & \checkmark & \checkmark & \xmark & \xmark & \xmark & \xmark & \checkmark & \xmark \\ \hline 
   
  \end{tabular}
\end{threeparttable}
 \label{tab:comp}
\end{center}
\scriptsize NOTE: N.A. is the abbreviation for Not Applicable; $^a$ indicates that the unforgeability implies the integrity of users' contact information and the prevention of false positives injection; $^a$ indicates that unforgeability is partially satisfied while not considering malicious colluding entities. 
\end{table*}

\section{SPOT Overview}
\label{sec:spec}

This section first presents the involved entities and gives an overview of \textsf{SPOT}. Then, it details the system model with its procedures and algorithms and defines the threat model through formal security games. 

\subsection{Entities}
\label{sec:sysmodel}

\begin{figure*}[!t]
\centering
\includegraphics[width=12cm,height=8.4cm]{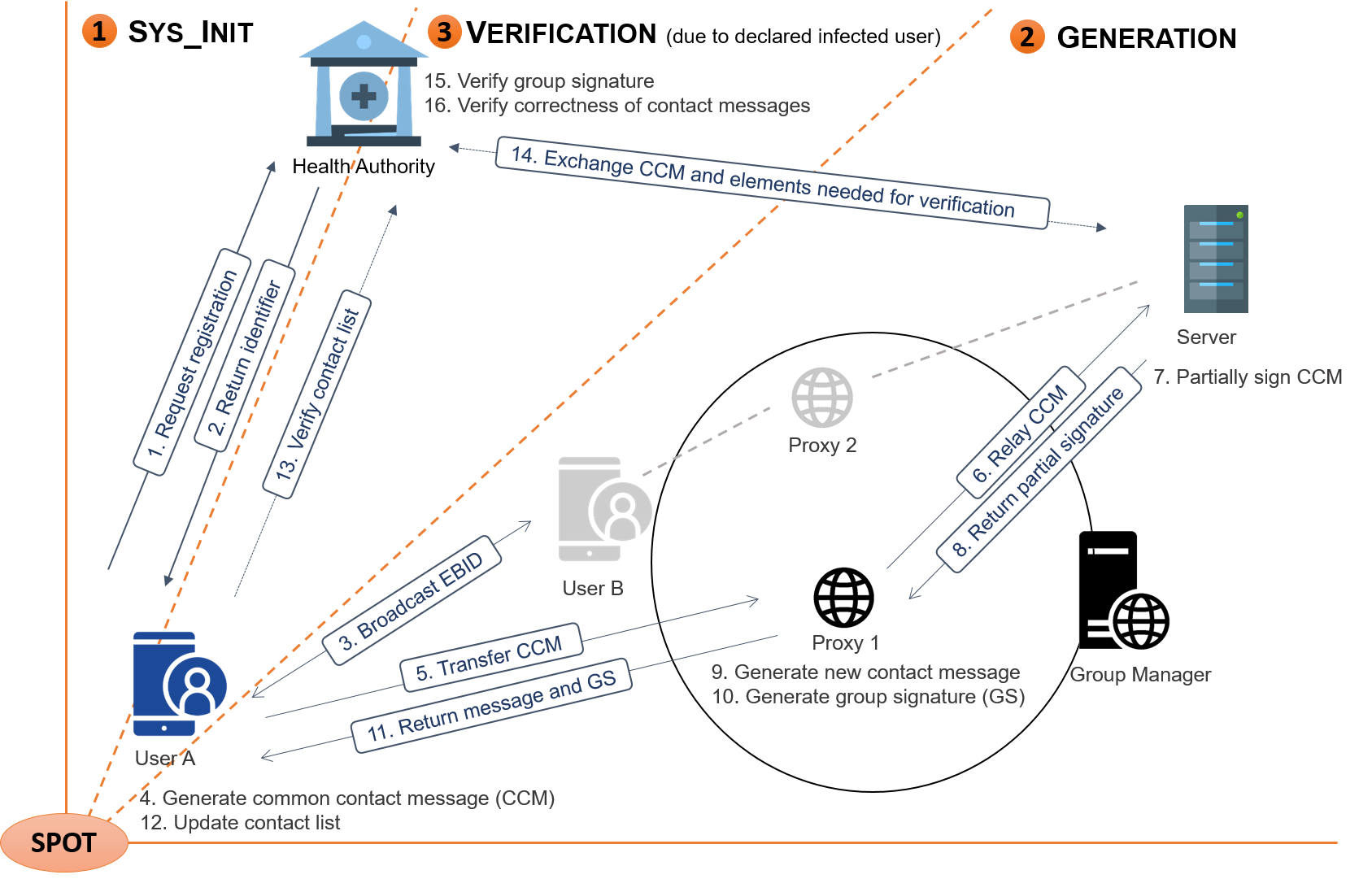}
\caption{Overview of the SPOT protocol}
\label{fig:globalarchi}
\end{figure*}

Figure \ref{fig:globalarchi} depicts the four main actors involved in \textsf{SPOT} with their interactions according to the different phases. Four actors are defined as follows:

\begin{itemize}
    \item The user ($\mathcal{U}$) represents the entity that owns the device where the proximity-tracing application is installed. 
   During the \textsc{Generation} phase, $\mathcal{U}$ broadcasts his EBID (Ephemeral Bluetooth IDentifier), collects the EBIDs of other users in proximity, and computes a common contact message shared between each two users being in contact. $\mathcal{U}$ wants to receive alerts if he was in contact with confirmed cases.
    \item The Health Authority ($\mathcal{HA}$) is responsible for issuing users' identifiers during the \textsc{Sys\_Init} phase, and for checking the correctness of the contact messages provided by an infected user during the \textsc{Verification} phase. 
    \item The Server ($\mathcal{S}$) is responsible for anonymously collecting and storing users' contact messages relayed by proxies during the \textsc{Generation} phase\footnote{The Server can be distributed by considering one or several servers per geographical area, each server participating in locally storing part of users' contact messages databases. All the parts are then collected on offline in a centralized server. Thus, for ease of presentation, we consider only one server.}. $\mathcal{S}$ performs a real-time verification of the received contacts during the \textsc{Generation} phase, in order to help $\mathcal{HA}$ to verify the correctness and integrity of the contact messages.
    \item The Proxy ($\mathcal{P}$) is considered as a member of a group of proxies managed by the group manager ($\mathcal{GM}$)\footnote{Proxies are distributed over several geographical areas. We assume that a load-balancing is established between at least two proxies in the same geographical area to ensure the system availability in case of failure or overload. More precisely, proxies in a same geographical area are separated into two subsets - a primary and a secondary - and two users in a contact interaction must contact proxies belonging to different subsets in order to prevent a proxy from gaining too much knowledge about users' interactions.}. Proxies form an intermediate layer by relaying the common contact messages of users to $\mathcal{S}$ in order to ensure the anonymity of involved users towards the server during the \textsc{Generation} phase. Proxies also play an important role in ensuring data integrity and user geolocation privacy thanks to group signatures. 
\end{itemize}

\subsection{Overview}

\textsf{SPOT} is set upon an hybrid architecture that leverages the best of the centralized and decentralized settings in proximity-tracing protocols. It relies on a proxy-based solution to preserve users' privacy (i.e., users remain anonymous towards the server, thus preventing users' tracking) and is based on a semi-trusted computing server to ensure data consistency and integrity (i.e., users are ensured that the received alerts are correct). The architecture of the proposed protocol is depicted in Figure \ref{fig:globalarchi}.
\textsf{SPOT} involves three main phases: \textsc{Sys\_Init}, \textsc{Generation} and \textsc{Verification} presented hereafter.

The \textsc{Sys\_Init} phase consists of initializing the whole system. It relies on seven algorithms, referred to as $\mathsf{Set\_params}$, $\mathsf{HA\_keygen}$, $\mathsf{S\_keygen}$, $\mathsf{Setup\_ProxyGr}_{\mathcal{GM}}$ and $\mathsf{Join\_ProxyGr}_{\mathcal{P}/\mathcal{GM}}$, $\mathsf{Set\_UserID}_{\mathcal{HA}}$ and $\mathsf{Userkeygen}_{\mathcal{U}}$. During the \textsc{Sys\_Init} phase, a trusted authority\footnote{For ease of presentation, the trusted authority is neither presented in Figure \ref{fig:globalarchi} nor in the system's model entities.} generates the system public parameters published to all involved entities and the pair of keys of both $\mathcal{HA}$ and $\mathcal{S}$, relying on $\mathsf{Set\_params}$, $\mathsf{HA\_keygen}$ and $\mathsf{S\_keygen}$ algorithms. During this phase, the group manager defines the group of proxies. It generates the group signature parameters using the $\mathsf{Setup\_ProxyGr}_{\mathcal{GM}}$ algorithm and it interacts with each group member to derive the associated keys relying on the $\mathsf{Join\_ProxyGr}_{\mathcal{P}/\mathcal{GM}}$ algorithm. The Health Authority is also involved in this phase to register a user when installing the proximity-tracing application. $\mathcal{HA}$ generates a specific secret value $t_\mathcal{U}$ (only known by $\mathcal{HA}$) and a unique identifier $\mathtt{ID}_\mathcal{U}$ for each user ($\mathcal{U}$), using the $\mathsf{Set\_UserID}_{\mathcal{HA}}$ algorithm. Finally, $\mathcal{U}$ uses his identifier to generate his pair of keys relying on the $\mathsf{Userkeygen_{U}}$ algorithm. The user's identifier $\mathtt{ID}_\mathcal{U}$, secret value $t_\mathcal{U}$ and public key are stored in a database $DB_{USER}$ owned by $\mathcal{HA}$. We note that the trusted authority, the group manager and proxies are involved only once in the \textsc{Sys\_Init} phase, while the health authority must intervene every time a user wants to register.

The \textsc{Generation} phase occurs when two users $\mathcal{U}_A$ and $\mathcal{U}_B$ are in contact. It represents the process of generating contact messages and contact lists for users. Three main entities participate in this phase relying on three different algorithms, referred to as $\mathsf{Set\_CCM}_{\mathcal{U}}$, $\mathsf{S\_PSign}_{\mathcal{S}}$ and $\mathsf{P\_Sign}_{\mathcal{P}}$. At first, $\mathcal{U}_A$ and $\mathcal{U}_B$ execute the $\mathsf{Set\_CCM}_{\mathcal{U}}$ algorithm to generate a common contact message relying on their random $EBID$s (denoted by $\mathtt{D}^e_A$ and $\mathtt{D}^e_B$) for an epoch $e$\footnote{An epoch $e$ denotes a period of time in which the Bluetooth identifier (EBID) remains unchanged.}. 
$\mathcal{U}_A$ and $\mathcal{U}_B$ choose two different proxies to relay their common contact message to the server. For this purpose, they compare their $EBID$s, i.e., if $\mathtt{D}^e_A$ $>$ $\mathtt{D}^e_B$, $\mathcal{U}_A$ chooses the first proxy and $\mathcal{U}_B$ selects the second one, and vice versa. 
Each of the two proxies relays the common contact message to the server. 
$\mathcal{S}$ checks if the two copies are similar. 
If so, $\mathcal{S}$ executes the $\mathsf{S\_PSign}_{\mathcal{S}}$ algorithm to partially sign the common contact message, thus proving that the contact message correctly reached the server. 
Afterwards, given back only a correct message, each proxy executes the $\mathsf{P\_Sign}_{\mathcal{P}}$ algorithm.
Indeed, each proxy extends the message, given by $\mathcal{S}$, with the corresponding user's identifier and it signs the resulting message on behalf of the group.
He, finally, sends back the message and the corresponding group signature to the user and closes the communication session, while removing all the exchanged and generated contact information. The user adds the group signature, along with the common contact message, the date, time and duration of contact, in his contact list $CL_\mathcal{U}$. Note that each contact information is stored for $\Delta$ days.

The \textsc{Verification} phase is run by the health authority to check the correctness of a contact list $CL_\mathcal{U}$ provided by $\mathcal{U}$ during a period of time $t$. To this end, $\mathcal{HA}$ performs three successive verifications relying on two main algorithms, referred to as $\mathsf{Sig\_Verify}_{\mathcal{HA}}$ and $\mathsf{CCM\_Verify}_{\mathcal{HA}}$. (i) $\mathcal{HA}$ checks if, in his $DB_{USER}$ database, $\mathcal{U}$ is infected\footnote{We suppose that the health status of a user is updated when being tested. Indeed, to be tested, $\mathcal{U}$ has to provide an encrypted form of his identifier $\mathtt{ID}_\mathcal{U}$  (i.e., $\mathtt{ID}_\mathcal{U}$ is encrypted meaning the $\mathcal{HA}$ public key). Afterwards, the analysis' result is sent with the encrypted identifier to $\mathcal{HA}$, that extracts the identifier and updates the user's health status in the $DB_{USER}$ database.}. (ii) $\mathcal{HA}$ checks the validity of the group signatures relying on the $\mathsf{Sig\_Verify}_{\mathcal{HA}}$ algorithm, w.r.t. the messages contained in the contact list $CL_{\mathcal{U}}$. (iii) $\mathcal{HA}$ verifies that the contact messages have been correctly generated and have successfully reached $\mathcal{S}$, using the $\mathsf{CCM\_Verify}_{\mathcal{HA}}$ algorithm. 

It is worth mentioning that if one of the verifications given above fails, the contact message is rejected. Otherwise, $\mathcal{HA}$ collects all verified messages of all infected users in a set $S_{\mathtt{CCM}}$ that she signs. Note that for each period of time $t$, $\mathcal{HA}$ removes users' contact lists after verifications. $S_{\mathtt{CCM}}$ and the corresponding signature are sent to the server that shares them with all users. To compute the risk score, each user compares the set $S_{\mathtt{CCM}}$ with his contact list, taking into account the number of infected users being contacted and the contact duration.

For ease of presentation, the different notations used in this paper are depicted in Table~\ref{tab:nota}.

\begin{table}[!ht]
\begin{center}
\caption{Notations used in this paper}
\resizebox{\columnwidth}{!}{
\begin{tabular}{c|c}
\hline
    Notation & Description \\ \hline
    \hline
    $\mathcal{U}$  & User \\ \hline
    $\mathcal{HA}$  & Health Authority \\ \hline
    $\mathcal{S}$ & Server \\ \hline
    $\mathcal{P}$  & Proxy \\ \hline
    $\mathcal{GM}$  & Group Manager \\ \hline
    $\mathcal{TA}$ & Trusted Authority \\ \hline 
    $\mathtt{ID}_{\mathcal{U}}$ & An identifier of a user $\mathcal{U}$ \\ \hline
    $t_{\mathcal{U}}$ & A secret value associated to $\mathcal{U}$ \\ \hline
    $DB_{USER}$ & The users' database at $\mathcal{HA}$ \\ \hline
    $\mathtt{D}^e$ & An ephemeral Bluetooth Identifier during an epoch $e$ \\ \hline
    $CL_{\mathcal{U}}$ & A user's contact list \\ \hline
    $\lambda$ & A security parameter \\ \hline
    $pp$ & The system's public parameters \\ \hline
    $\mathtt{sk}$ & A private key \\ \hline
    $\mathtt{pk}$ & A public key \\ \hline 
    $\mathtt{vk}_g$ & The group public parameters \\ \hline 
    $\mathtt{CCM}$ & A common contact message \\ \hline 
    $(\mathtt{PS}, \mathtt{PS}')$ & A partial signature \\ \hline 
    $\sigma$ & A signature \\ \hline 
    $\mathtt{M}$ & A message derived from $\mathtt{PS}$ and $\mathtt{ID}_{\mathcal{U}}$ \\ \hline 
    $\pi$ & A NIWI proof \\ \hline
    $S_{\mathtt{CCM}}$ & A set of verified $\mathtt{CMM}$s of infected users \\ \hline
  \end{tabular}}
 \label{tab:nota}
\end{center}
\end{table}

\subsection{System Model}

Based on the three phases, Figure \ref{fig:seqdiag} presents the chronological sequence of twelve PPT algorithms, defined below. For ease of presentation, we consider only one user and one proxy in the sequence diagram. For the \textsc{Generation} phase, we suppose that two users have been in contact and exchanged their EBIDs. The \textsc{Verification} phase occurs only if the user receives a negative analysis' result.

\begin{figure*}[!t]
\centering
\includegraphics[ width=14cm,height=9.7cm]{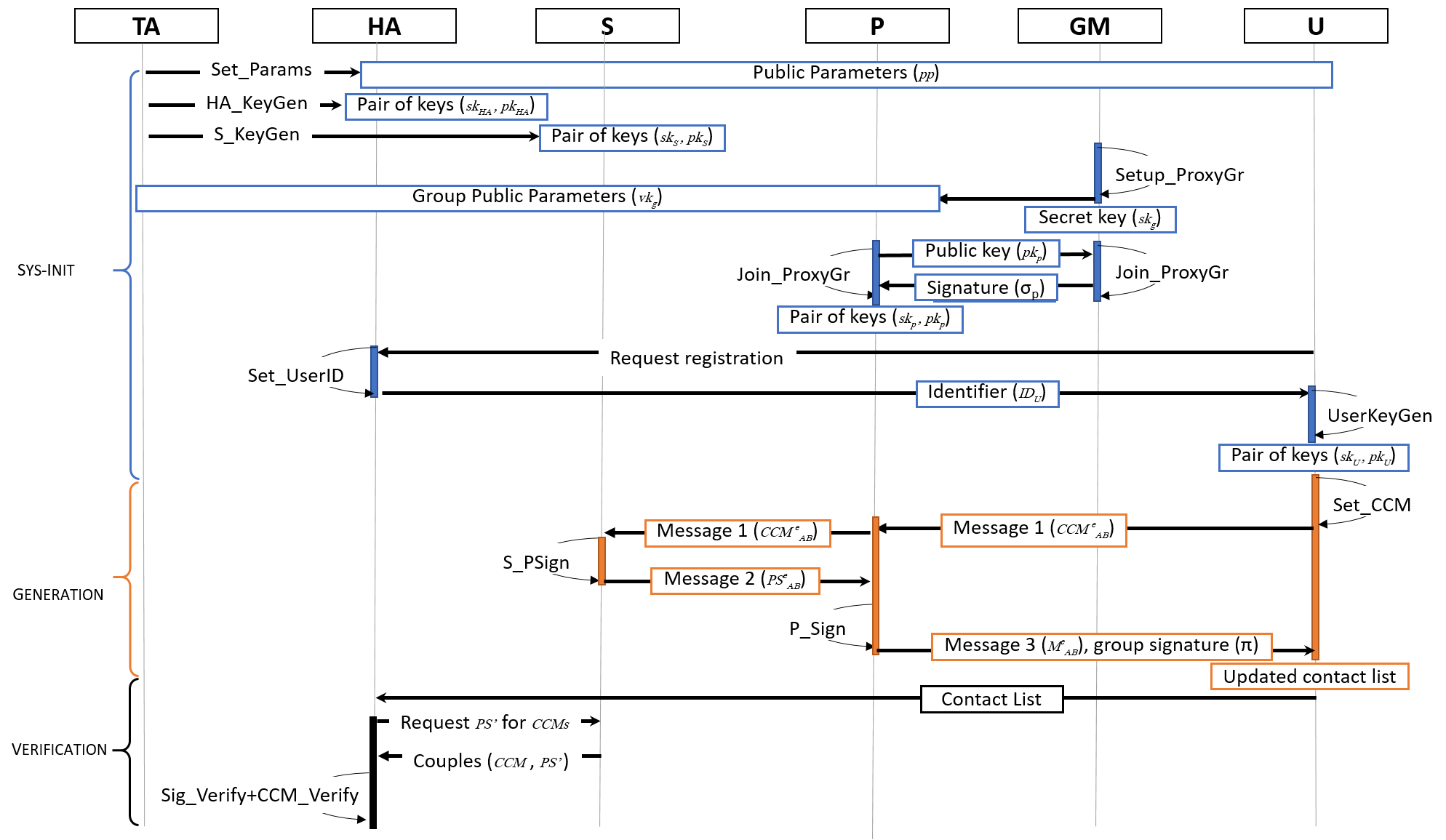}
\caption{Workflow of the SPOT protocol}
\label{fig:seqdiag}
\end{figure*}

\begin{itemize}
    \item \textsc{Sys\_Init} phase:\\
    $\mathsf{Set\_params}(\lambda) \rightarrow pp$ -- run by a trusted authority. Given the security parameter $\lambda$, this algorithm generates the system public parameters $pp$ that will be considered as a default input for all the following algorithms.
    
    $\mathsf{K\_keygen}() \rightarrow (\mathtt{sk}_{j}, \mathtt{pk}_{j})$ -- performed by a trusted authority. It returns the pair of keys $(\mathtt{sk}_{j}, \mathtt{pk}_{j})$ of $j$ where $j = \{\mathcal{HA}, \mathcal{S}\}$.

    $\mathsf{Setup\_ProxyGr}_{\mathcal{GM}}() \rightarrow (\mathtt{sk}_{g}, \mathtt{vk}_{g})$ -- this algorithm is performed by the group manager to set up the group signature. It returns the proxies' group verification key $\mathtt{vk}_g$ represented as $(\mathtt{pk}_g, \Sigma_{\mathsf{NIWI}})$, where $\mathtt{pk}_g$ is the public key of the group manager and $\Sigma_{\mathsf{NIWI}}$ is the Common Reference String CRS of a NIWI proof \cite{grothsahai}. The $\mathsf{Setup\_ProxyGr}_{\mathcal{GM}}$ algorithm also outputs the secret key $\mathtt{sk}_{g}$ that is privately stored by $\mathcal{GM}$.   
    
    $\mathsf{Join\_ProxyGr_{\mathcal{P}/\mathcal{GM}}}(\mathtt{sk}_{g}) \rightarrow (\mathtt{sk}_{p}, \mathtt{pk}_{p}, \sigma_p)$ -- this algorithm is performed through an interactive session between the proxy and the group manager. It takes as input the secret key $\mathtt{sk}_{g}$, and outputs the pair of keys $(\mathtt{sk}_{p}, \mathtt{pk}_{p})$ of $\mathcal{P}$ belonging to the group (i.e., $\mathcal{P}$ is responsible for generating his pair of keys), and a signature $\sigma_p$ over $\mathcal{P}$'s public key $\mathtt{pk}_p$ (i.e., $\sigma_p$ is generated by $\mathcal{GM}$).
    
    $\mathsf{Set\_UserID_{\mathcal{HA}}}() \rightarrow (t_\mathcal{U}, \mathtt{ID}_\mathcal{U})$ -- this algorithm is run by $\mathcal{HA}$ and returns a secret value $t_\mathcal{U}$ specific for $\mathcal{U}$ and the identifier $\mathtt{ID}_\mathcal{U}$ of $\mathcal{U}$.

    $\mathsf{Userkeygen_{U}}(\mathtt{ID}_\mathcal{U}) \rightarrow (\mathtt{sk}_{\mathcal{U}}, \mathtt{pk}_{\mathcal{U}})$ -- performed by $\mathcal{U}$ to set his pair of keys $(\mathtt{sk}_{\mathcal{U}}, \mathtt{pk}_{\mathcal{U}})$ relying on the identifier $\mathtt{ID}_\mathcal{U}$.
    
    \item \textsc{Generation} phase:\\
    $\mathsf{Set\_CCM_{\mathcal{U}}}(\mathtt{D}^e_A, \mathtt{D}^e_B) \rightarrow \mathtt{CCM}^e_{AB}$ -- run by each of two users $\mathcal{U}_A$ and $\mathcal{U}_B$ being in contact during an epoch $e$. Given two Bluetooth identifiers $\mathtt{D}^e_A$ and $\mathtt{D}^e_B$,
    this algorithm generates a common contact message $\mathtt{CCM}^e_{AB}$.

   $\mathsf{S\_PSign_{\mathcal{S}}}(\mathtt{CCM}^e_{AB}, \mathtt{sk}_{\mathcal{S}}) \rightarrow (\mathtt{PS}^e_{AB}, \mathtt{PS'}^e_{AB})$ -- run by $\mathcal{S}$. Given a common contact message $\mathtt{CCM}^e_{AB}$ sent by $\mathcal{U}_A$ and $\mathcal{U}_B$ through two different proxies $\mathcal{P}_1$ and $\mathcal{P}_2$, this algorithm outputs the couple ($\mathtt{PS}^e_{AB}$, $\mathtt{PS'}^e_{AB}$) that is stored with $\mathtt{CCM}^e_{AB}$ at $\mathcal{S}$, for $\Delta$ days. Note that only $\mathtt{PS}^e_{AB}$ is given back to $\mathcal{P}_1$ and $\mathcal{P}_2$ to prove that $\mathtt{CCM}^e_{AB}$ has been successfully received and verified by $\mathcal{S}$ (i.e., a real contact took place), while $\mathtt{PS'}^e_{AB}$ is kept secret at $\mathcal{S}$ and is sent only to $\mathcal{HA}$ to check the correctness of a contact message provided by a infected user.
    
    $\mathsf{P\_Sign}_{\mathcal{P}}(\mathtt{vk}_g, \mathtt{sk}_p, \mathtt{pk}_p, \sigma_p, \mathtt{ID}_{\mathcal{U}_A}, \mathtt{PS}^e_{AB}) \rightarrow (\mathtt{M}^e_{AB}$, $\sigma_m$, $\pi)$\footnote{In this algorithm, we only consider user $\mathcal{U}_A$ with $\mathtt{ID}_{\mathcal{U}_A}$. The same operations are performed for user $\mathcal{U}_B$ with $\mathtt{ID}_{\mathcal{U}_B}$.} -- performed by the proxy $\mathcal{P}$ ($\mathcal{P}_1$ or $\mathcal{P}_2$). This algorithm takes as input the proxies' group public parameters $\mathtt{vk}_g$, the pair of keys $(\mathtt{sk}_p, \mathtt{pk}_p)$ of $\mathcal{P}$, the signature $\sigma_p$ over $\mathcal{P}$'s public key, the identifier $\mathtt{ID}_{\mathcal{U}_A}$ of user $\mathcal{U}_A$ and the message $\mathtt{PS}^e_{AB}$. 
    It returns a signature $\sigma_m$ over a new message $\mathtt{M}^e_{AB}$ and a group signature represented by a NIWI proof $\pi$ over the two signatures $\sigma_p$ and $\sigma_m$. The couple ($\mathtt{M}^e_{AB}$, $\pi$) is sent to user $\mathcal{U}_A$ to be stored with the contact message $\mathtt{CCM}^e_{AB}$ in his contact list. Note that each input of the contact list is stored for $\Delta$ days.

    \item \textsc{Verification} phase:\\
    $\mathsf{Sig\_Verify_{\mathcal{HA}}} (\mathtt{vk}_g, \mathtt{M}^e_{AB}, \pi) \rightarrow b$ -- performed by $\mathcal{HA}$. Given the public parameters $\mathtt{vk}_g$, a message $\mathtt{M}^e_{AB}$ from the contact list of an infected user, and the corresponding NIWI proof $\pi$, the $\mathsf{Sig\_Verify_{\mathcal{HA}}}$ algorithm returns a bit $b \in \{0,1\}$ stating whether the proof is valid. 
    $\mathsf{CCM\_Verify_{\mathcal{HA}}}(\mathtt{M}^e_{AB}, \mathtt{PS'}^e_{AB}, \mathtt{pk}_{\mathcal{S}}, t_{\mathcal{U}_A}) \rightarrow b$ -- run by $\mathcal{HA}$. This algorithm takes as input the message $\mathtt{M}^e_{AB}$, the message $\mathtt{PS'}^e_{AB}$ requested from $\mathcal{S}$, the server's public key $\mathtt{pk}_{\mathcal{S}}$ and the secret value $t_{\mathcal{U}_A}$, and outputs a bit $b \in \{0,1\}$, i.e., $\mathtt{CCM}^e_{AB}$ is correctly generated or not.
\end{itemize}

\subsection{Threat Model}
\label{sec:thrmodel}

In this section, we first present the adversaries considered in \textsf{SPOT}, and then, the formal definitions of the different security and privacy properties. 

\begin{itemize}
    \item {\textbf{A malicious user ($\mathcal{U}$)}}: this adversary attempts to inject false contact messages or contact messages of other users in his contact list. He may also collude with corrupted proxies or malicious users.
\item{\textbf{A honest but curious health authority ($\mathcal{HA}$)}}: given a valid group signature, $\mathcal{HA}$ tries to identify the signer (i.e., proxy) of a particular message, hence for identifying the appropriate geographical area and for tracking the user's movements. She may also attempt to link two signatures issued by the same group member. A curious $\mathcal{HA}$ may also try to identify, from a contact list of a particular user, the list of users he had been in contact with.   
\item{\textbf{A honest but curious server ($\mathcal{S}$)}}: he attempts to link several common contact messages generated by the same user, to trace users' movements.
 \item \textbf{A malicious proxy ($\mathcal{P}$)}: this adversary, either colluding with a malicious user or with $n$ other proxies, attempts to forge the partial signature of the server and to generate a valid signature over a false contact message. 
\end{itemize}
Both malicious users and malicious proxies are considered against security properties, i.e., unforgeability, anti-replay, while the curious health authority and server are considered against privacy requirements, i.e., unlinkability and anonymity. 
The different adversaries are involved in different phases.\\ 
Note that the \textbf{anti-replay} property which aims at mitigating the multi-submission of the same contact information is not formally presented below, but is informally discussed in Section \ref{sec:secanal}.
The following properties are defined w.r.t the corresponding phases and the involved adversaries.

\begin{remark}
We do not deeply analyze the case of a malicious $\mathcal{GM}$ although our scheme is resistant against this adversary.
Indeed, proxies are responsible for generating their key pair and only their public keys are shared with $\mathcal{GM}$. 
Thus, unless holding a proxy's secret key, $\mathcal{GM}$ is not able to generate a valid signature on behalf of $\mathcal{P}$ thanks to the unforgeability of the signature scheme. 
\end{remark}

\subsubsection{Unforgeability}
\label{sec:unforg}

The unforgeability property ensures the security of \textsf{SPOT} for the different phases. It states that a malicious user is not able to forge his contact list (i.e., forging either the group signature or the server's partial signature when colluding with a malicious proxy)\footnote{We assume that (i) malicious user refers to either a single user or colluding users, and (ii) the group signature scheme used in \textsf{SPOT} is unforgeable as proven in \cite{XCGSIG}, thus in the security model and analysis, we will only consider the unforgeability of the server's partial signature.}. Formally, this is defined in a game $\textbf{Exp}_{\mathcal{A}}^{unforg}$ where an adversary $\mathcal{A}$, playing the role of a corrupted proxy colluding with a malicious user, has access to a $\mathsf{S\_PSign}$ oracle. We note that, for each session $i$, $\mathcal{A}$ only gets $\mathtt{PS}^i$ from the $\mathsf{S\_PSign}$ oracle, while $\mathtt{PS'}^i$ is kept secret from the adversary. Then, given a valid message $\mathtt{PS'}$ that cannot be obtained by combining either a part of or all messages $\mathtt{PS}^i$, $\mathcal{A}$ succeeds if it outputs a valid message $\mathtt{PS}^*$ to be signed using $\mathsf{P\_Sign}$, such that the $\mathsf{CCM\_Verify}$ verification holds.

\begin{definition}{\textbf{Unforgeability}} -- We say that \textsf{SPOT} satisfies the unforgeability property, if for every \emph{PPT} adversary $\mathcal{A}$, there exists a negligible function $\kappa$ such that: $Pr[\textbf{Exp}_{\mathcal{A}}^{unforg}(1^{\lambda})=1] \leq \kappa(\lambda)$, where $\textbf{Exp}_{\mathcal{A}}^{unforg}$ is given below. 

\label{def:unforg}
\end{definition}

\noindent
\fbox{
\begin{minipage}[r][\dimexpr 0.330\textheight-2\fboxsep-2\fboxrule\relax][t]{\dimexpr 0.435 \textwidth-2\fboxsep-2\fboxrule\relax}
    $\textbf{\textit{Exp}}_{\mathcal{A}}^{unforg}(\lambda)$ \\
    $pp \leftarrow \mathsf{Set\_params}(\lambda) $ \\
    $(\mathtt{sk_{\mathcal{HA}}}, \mathtt{pk_{\mathcal{HA}}}) \leftarrow \mathsf{HA\_keygen}(pp)$ \\
    $(\mathtt{sk_{\mathcal{S}}}, \mathtt{pk_{\mathcal{S}}}) \leftarrow \mathsf{S\_keygen}(pp) $ \\
    $(\mathtt{sk}_{g}, \mathtt{vk}_{g}) \leftarrow \mathsf{Setup\_ProxyGr}(pp)$ \\
    $(\mathtt{sk}_{p}, \mathtt{pk}_{p}, \sigma_p) \leftarrow \mathsf{Join\_ProxyGr}(pp, \mathtt{sk}_{g})$ \\
    $\mathtt{ID_{\mathcal{U}}} \leftarrow \mathsf{Set\_UserID}(pp)$ \\
    $(\mathtt{sk}_{\mathcal{U}}, \mathtt{pk}_{\mathcal{U}}) \leftarrow \mathsf{Userkeygen}(pp, \mathtt{ID_{\mathcal{U}}})$ \\
    $\mathtt{CCM} \leftarrow \mathsf{Set\_CCM}(\mathtt{D}_1, \mathtt{D}_2) $ \\
    $(\mathtt{PS}, \mathtt{PS'})$ $\leftarrow$ $\{\mathsf{S\_PSign}(\mathtt{CCM}, \mathtt{sk_{\mathcal{S}}})\}$ \\
    $\mathcal{O}$ $\leftarrow$ $\{\mathsf{S\_PSign}(\cdot, \mathtt{sk_{\mathcal{S}}})\}$ \\
    $\mathtt{PS}^* \leftarrow \mathcal{A}^{\mathcal{O}} (\mathtt{vk_{g}}, \mathtt{sk_{p}}, \mathtt{pk_{p}}, \sigma_p, ID_{\mathcal{U}}, pp, \mathtt{PS'}) $ 
    
    $\quad$ \textit{letting}  $\mathtt{CCM}$  \textit{and} $\mathtt{PS}^i$ \textit{denote the queries} 
    
    $\quad$ \textit{and answers to and from oracle} $\mathsf{S\_PSign}$ \\
    $(\mathtt{M}^*, \sigma^*, \pi^*) \leftarrow  \mathsf{P\_Sign}(\mathtt{vk}_g, \mathtt{sk}_p, \mathtt{pk}_p, \sigma_p, ID_{\mathcal{U}}, \mathtt{PS}^*)$ \\
    \textbf{If} $\mathsf{CCM\_Verify}( \mathtt{M}^*, \mathtt{PS'}^*, \mathtt{pk_{\mathcal{S}}}, t_{\mathcal{U}}) = 1 $
    
    
    
    $\quad$ \textbf{return} 1\\
    \textbf{Else return} 0
\end{minipage}
}

\subsubsection{Unlinkability}

The unlinkability property can be divided into two sub-properties. The first one constitutes the \emph{group-signature unlinkability} stating that a curious health authority is not able to link two or several group signatures issued by the same proxy during the \textsc{Verification} phase. The second sub-property \emph{multi-CCM unlinkability} ensures that a curious server is not able to link two or several common contact messages to the same user during the \textsc{Generation} phase
\footnote{The collusion between the health authority and the server does not pose additional and plausible threats to the different procedures of the whole framework. Indeed, during the \textsc{Generation} phase, contact messages are anonymous to the server (and a possible colluding health authority); during the \textsc{Verification} phase, the health authority knows the true identity of the confirmed cases with their contact information; as  such, a collusion between the server and the authority does not bring extra knowledge.}.

We note that the \emph{multi-CCM unlinkability} property will be informally
discussed in Section \ref{sec:secanal}. In this section, we only focus on the \emph{group-signature unlinkability}. Formally, this property is defined in a game $\textbf{Exp}_{\mathcal{A}}^{unlink}$ where an adversary $\mathcal{A}$ acting as a curious $\mathcal{HA}$ has access to a $\mathsf{P\_Sign}$ oracle. The adversary may query this oracle on the same message $\mathtt{PS}^*$ and on a tuple $((\mathtt{sk}_{p_j}, \mathtt{pk}_{p_j}, \sigma_{p_j})$, where $j \in \{0,1\}$ (i.e., the tuple belongs either to proxy $\mathcal{P}_0$ or proxy $\mathcal{P}_1$).
A left-or-right oracle $\mathsf{LoRSig}$ is initialized with a secret random bit $b$ and returns to $\mathcal{A}$ $\mathsf{P\_Sign}$ on message $\mathtt{PS}^*$ and respectively on tuples $(\mathtt{sk}_{p_0}, \mathtt{pk}_{p_0}, \sigma_{p_0})$ and $(\mathtt{sk}_{p_b}, \mathtt{pk}_{p_b}, \sigma_{p_b})$.
The adversary wins the game if he successfully predicts the bit $b$ (i.e., the guessing probability should be greater than $\frac{1}{2}$).

\begin{definition}{\textbf{Unlinkability}} -- We say that \textsf{SPOT} satisfies the unlinkability property, if for every \emph{PPT} adversary $\mathcal{A}$, there exists a negligible function $\kappa$ such that: $Pr[\textbf{Exp}_{\mathcal{A}}^{unlink}(\lambda)=1] = \frac{1}{2} \pm \kappa(\lambda)$, where $\textbf{Exp}_{\mathcal{A}}^{unlink}$ is defined below.
\label{def:unlink}
\end{definition}

\noindent
\fbox{%
\begin{minipage}[r][\dimexpr 0.315\textheight-2\fboxsep-2\fboxrule\relax][t]{\dimexpr 0.435\textwidth-2\fboxsep-2\fboxrule\relax}%
   $\textbf{\textit{Exp}}_{\mathcal{A}}^{unlink}(\lambda)$ \\
    $pp \leftarrow \mathsf{Set\_params}(\lambda) $ \\
    $(\mathtt{sk_{\mathcal{HA}}}, \mathtt{pk_{\mathcal{HA}}}) \leftarrow \mathsf{HA\_keygen}(pp)$ \\
    $(\mathtt{sk_{\mathcal{S}}}, \mathtt{pk_{\mathcal{S}}}) \leftarrow \mathsf{S\_keygen}(pp) $ \\
    $(\mathtt{sk}_{g}, \mathtt{vk}_{g}) \leftarrow \mathsf{Setup\_ProxyGr}(pp)$ \\
    $(\mathtt{sk}_{p_i}$, $\mathtt{pk}_{p_i}$, $\sigma_{p_i})$ $\leftarrow$ $\mathsf{Join\_ProxyGr}$ ($pp$, $\mathtt{sk}_{g}$), $i \in \{0,1\}$ \\
    $\mathtt{ID_{\mathcal{U}}} \leftarrow \mathsf{Set\_UserID}(pp)$ \\
    $(\mathtt{sk}_{\mathcal{U}}, \mathtt{pk}_{\mathcal{U}}) \leftarrow \mathsf{Userkeygen}(pp, \mathtt{ID_{\mathcal{U}}})$ \\
    $m^* \leftarrow \mathsf{Set\_CCM}(\mathtt{D}_{1}, \mathtt{D}_{2})$ \\
    $(\mathtt{PS}^*, \mathtt{PS'}^*)  \leftarrow \mathsf{S\_PSign}(m^*, \mathtt{sk_{S}})$ \\
    $b \leftarrow \{0,1\}$ \\
    $\mathcal{O}$ $\leftarrow$ $\{\mathsf{P\_Sign}$($\cdot$,$\mathtt{sk}_{p_j}$,$\mathtt{pk}_{p_j}$,$\sigma_{p_j}$,$\cdot$,$\cdot$), $\mathsf{LoRSig}$($\cdot$,$\cdot$,$b)\}$ \\
    $b'$ $\leftarrow$ $\mathcal{A}^{\mathcal{O}}$ $(\mathtt{sk_{\mathcal{HA}}}$,$\mathtt{pk_{\mathcal{HA}}}$, $\mathtt{vk}_{g}$,$\mathtt{pk}_{p_0}$,$\mathtt{pk}_{p_1}$,$pp$,$\mathtt{PS}^*) $ \\
    \textbf{If} $b = b'$ 
    
    $\quad$ \textbf{return} 1\\
    \textbf{Else return} 0
\end{minipage}%
}%
\medskip

\fbox{%
\begin{minipage}[r][\dimexpr 0.150\textheight-2\fboxsep-2\fboxrule\relax][t]{\dimexpr 0.435\textwidth-2\fboxsep-2\fboxrule\relax}%
    $\mathsf{LoRSig}$ $(\mathtt{vk}_{g}$, (($\mathtt{sk}_{p_0}$, $\mathtt{pk}_{p_0}$, $\sigma_{p_0})$, ($\mathtt{sk}_{p_1}$, $\mathtt{pk}_{p_1}$, $\sigma_{p_1}$)), $\mathtt{ID_{\mathcal{U}}}$, $\mathtt{PS}^*, b)$\\
    
     $({\mathtt{M}}^*, \sigma^*, \pi^*)$ $\leftarrow$ $\mathsf{P\_Sign}$($\mathtt{vk}_g$, $\mathtt{sk}_{p_0}$, $\mathtt{pk}_{p_0}$, $\sigma_{p_0}$, $\mathtt{ID_{\mathcal{U}}}$ , $\mathtt{PS}^*) $\\
     $({\mathtt{M}}^*, \sigma^*_b, \pi^*_b)$ $\leftarrow$ $\mathsf{P\_Sign}$($\mathtt{vk}_g$, $\mathtt{sk}_{p_b}$, $\mathtt{pk}_{p_b}$, $\sigma_{p_b}$, $\mathtt{ID_{\mathcal{U}}}$ , $\mathtt{PS}^*)  $\\
     \textbf{return} $(({\mathtt{M}}^*, \pi^*), ({\mathtt{M}}^*, \pi^*_b))$
\end{minipage}%
}%

\subsubsection{Anonymity}

This property guarantees that no entity is able to identify users involved in a contact list (i.e., the owner and the contacted users), during the \textsc{Verification} phase, and is described through the game $\textbf{Exp}_{\mathcal{A}}^{anon}$. The anonymity property implies that even if $\mathcal{HA}$ knows that a contact list belongs to a user ($\mathcal{U}$), $\mathcal{HA}$ is not able to identify users being in contact with $\mathcal{U}$
\footnote{We assume that the probability of two confirmed users being in contact and submitting their respective contact lists to $\mathcal{HA}$ at the same period, is low.}. This should hold even if an efficient adversary, playing the role of the curious health authority, is given access to $\mathsf{Set\_CCM}$, $\mathsf{S\_PSign}$, $\mathsf{P\_Sign}$ oracles. $\mathcal{A}$ can learn contact messages and signatures associated to the selected users' identifiers. $\mathcal{A}$ also gets access to a left-or-right oracle $\mathsf{LoRCU}$ which is initialized with a secret random bit $b \in \{0,1\}$. $\mathcal{A}$ may query this oracle on $\mathtt{ID_{\mathcal{U}_0}}$ and $\mathtt{ID_{\mathcal{U}_1}}$ referred to as the identifiers of respectively user $\mathcal{U}_0$ and user $\mathcal{U}_1$. Observe that user $\mathcal{U}_A$ is involved in all queries.
$\mathtt{D}^*_{\mathcal{U}_A}$ and $\mathtt{D}^*_{\mathcal{U}_b}$, respectively belonging to user $\mathcal{U}_A$ and user $\mathcal{U}_b$, are randomly selected in order to execute the $\mathsf{LoRCU}$ oracle.
To win the proposed anonymity game, the adversary should predict the bit $b$ (i.e., which one of users $\mathcal{U}_0$ and $\mathcal{U}_1$ is involved in the contact with user $\mathcal{U}_A$) with a probability greater than $\frac{1}{2}$.

\begin{definition}{\textbf{Anonymity}} -- We say that \textsf{SPOT} fulfills the anonymity requirement, if for every \emph{PPT} adversary $\mathcal{A}$, there exists a negligible function $\kappa$ such that: $Pr[\textbf{Exp}_{\mathcal{A}}^{anon}(1^{\lambda})=1] = \frac{1}{2} \pm \kappa(\lambda)$, where $\textbf{Exp}_{\mathcal{A}}^{anon}$ is defined as follows.
\label{def:anon}
\end{definition}

\noindent
\fbox{%
\begin{minipage}[l][\dimexpr 0.310\textheight-2\fboxsep-2\fboxrule\relax][t]{\dimexpr 0.435\textwidth-2\fboxsep-2\fboxrule\relax}%
   $\textbf{\textit{Exp}}_{\mathcal{A}}^{anon}(\lambda)$ \\
    $pp \leftarrow \mathsf{Set\_params}(\lambda) $ \\
    $(\mathtt{sk_{\mathcal{HA}}}, \mathtt{pk_{\mathcal{HA}}})$ $\leftarrow$ $\mathsf{HA\_keygen}(pp)$ \\
    $(\mathtt{sk_{\mathcal{S}}}, \mathtt{pk_{\mathcal{S}}})$ $\leftarrow$ $\mathsf{S\_keygen}(pp) $ \\
    $(\mathtt{sk}_{g}, \mathtt{vk}_{g})$ $\leftarrow$ $\mathsf{Setup\_ProxyGr}(pp)$ \\
    $(\mathtt{sk}_{p}, \mathtt{pk}_{p}, \sigma_p)$ $\leftarrow$ $\mathsf{Join\_ProxyGr}(pp, \mathtt{sk}_{g})$ \\
    $\mathtt{ID}_{\mathcal{U}_A}$ $\leftarrow$ $\mathsf{Set\_UserID}(pp)$  \\
    $(\mathtt{sk}_{\mathcal{U}_A}, \mathtt{pk}_{\mathcal{U}_A})$ $\leftarrow$ $\mathsf{Userkeygen}(pp, \mathtt{ID}_{\mathcal{U}_A})$  \\
    $\mathtt{ID}_{\mathcal{U}_i}$ $\leftarrow$ $\mathsf{Set\_UserID}(pp)$, $i \in \{0..N\}$ \\
    $(\mathtt{sk}_{\mathcal{U}_i}, \mathtt{pk}_{\mathcal{U}_i})$ $\leftarrow$ $\mathsf{Userkeygen}(pp, \mathtt{ID}_{\mathcal{U}_i})$, $i \in \{0..N\}$ \\
    $b \leftarrow \{0,1\}$ \\
    $\mathcal{O}$ $\leftarrow$ $\{\mathsf{Set\_CCM}$($\cdot$,$\cdot$), $\mathsf{S\_PSign}$($\cdot$,$\mathtt{sk_{\mathcal{S}}}$),  $\mathsf{P\_Sign}$($\cdot$,$\mathtt{sk_{p}}$, $\cdot$,$\sigma_p$,$\cdot$,$\cdot$), $\mathsf{LoRCU}$($\cdot$,$\cdot$,$b$,$\cdot)$ \\
    $b'$ $\leftarrow$ $\mathcal{A}^{\mathcal{O}}$ $(\mathtt{sk_{\mathcal{HA}}}$, $\mathtt{pk_{\mathcal{HA}}}$, $pp$, $\mathtt{ID}_{\mathcal{U}_A}$, $\{\mathtt{ID}_{\mathcal{U}_i}\}_{i=0}^N) $ \\
    \textbf{If} $b = b'$ 
    
    $\quad$ \textbf{return} 1\\
    \textbf{Else return} 0
\end{minipage}%
}%

\medskip

\fbox{%
\begin{minipage}[r][\dimexpr 0.150\textheight-2\fboxsep-2\fboxrule\relax][t]{\dimexpr 0.435\textwidth-2\fboxsep-2\fboxrule\relax}%
   $\mathsf{LoRCU}$ $(\mathtt{D}^*_{\mathcal{U}_A}$, $\mathtt{D}^*_{\mathcal{U}_b}$, $b$, $\mathtt{vk}_{g}$, $\mathtt{sk_{\mathcal{S}}}$, $\mathtt{sk}_{p}$, $\mathtt{pk}_{p}$, $\sigma_p$, $\mathtt{ID}_{\mathcal{U}_A}$, $\mathtt{ID}_{\mathcal{U}_b})$\\
       
     $ {\mathsf{CCM}}^*_{b}$ $\leftarrow$ $\mathsf{Set\_CCM_{\mathcal{U}_A}}$ $(\mathtt{D}^*_{\mathcal{U}_A}$, $\mathtt{D}^*_{\mathcal{U}_b})$ \\
     $ ({\mathsf{PS}}^*_b, {\mathsf{PS}'}^*_b)$ $\leftarrow$ $\mathsf{S\_PSign}$ $({\mathsf{CCM}}^*_{b}$, $\mathtt{sk_{\mathcal{S}}})$ \\
     $(\mathtt{M}^*_b, \sigma^*_b, \pi^*_b)$ $\leftarrow$ $\mathsf{P\_Sign}$ $(\mathtt{vk}_g$, $\mathtt{sk}_p$, $\mathtt{pk}_p$, $\sigma_p$, $\mathtt{ID}_{\mathcal{U}_A}$, ${\mathsf{PS}}^*_b)$ \\
     \textbf{return} $({\mathsf{CCM}}^*_{b}, \mathtt{M}^*_b, \pi^*_b)$
\end{minipage}%
}%

\section{Building Blocks}
\label{sec:blocks}

After introducing bilinear maps and security standard assumptions in Section \ref{subsec:math}, the section next presents structure-preserving signatures \cite{XCGSIG} with their different variants as main building blocks of the \textsf{SPOT} protocol. 
Sections \ref{CSIG} and \ref{XSIG} describe respectively constant-size signatures and signatures on mixed-group messages that are instantiated in Appendix \ref{sec:gs} to build a group signature scheme on group element messages. 

\subsection{Mathematical Background and Cryptographic Assumptions}
\label{subsec:math}

Hereafter, we define bilinear maps and we present the computational indistinguishability property and the CDH assumption.

\subsubsection{Bilinear Maps}

Let $\mathbb{G}_1$=$\langle{g_1}\rangle$ and $\mathbb{G}_2$=$\langle{g_2}\rangle$ be two cyclic groups of order $n$ so there exists a bilinear map $e:\mathbb{G}_1 \times \mathbb{G}_2 \rightarrow \mathbb{G}_3$ that satisfies the following properties: (i) bilinearity for all $g_1 \in \mathbb{G}_1$, $g_2 \in \mathbb{G}_2$, (ii) non-degeneracy: $e(g_1,g_2) \neq 1$ and (iii) $e(g_1,g_2)$ is efficiently computable for any $g_1 \in \mathbb{G}_1$ and $g_2 \in \mathbb{G}_2$.

\subsubsection{Computational Witness-Indistinguishability}

The Computational Witness-Indistinguishability property is defined as follows: Let $L \in \mathcal{NP}$ be a language and let $(\mathcal{P}, \mathcal{V})$ be an interactive proof system for $L$. We say that $(\mathcal{P}, \mathcal{V})$ is \emph{witness-indistinguishable} (WI) if for every \emph{PPT} algorithm $\mathcal{V}^*$ and every two sequences $\{w^1_x\}_{x \in L}$ and $\{w^2_x\}_{x \in L}$ such that $w^1_x$ and $w^2_x$ are both witnesses for $x$, the following ensembles are computationally indistinguishable, where $z$ is an auxiliary input to $\mathcal{V}^*$:
\begin{enumerate}
    \item $\{\langle \mathcal{P}(w^1_x), \mathcal{V}^*\rangle (x)\}_{x \in L, z \in \{0,1\}^*}$
    \item $\{\langle \mathcal{P}(w^2_x), \mathcal{V}^*\rangle (x)\}_{x \in L, z \in \{0,1\}^*}$
\end{enumerate}

\subsubsection{Computational Diffie Hellman Assumption (CDH)} 
The CDH assumption is defined as follows: Let $\mathbb{G}$ be a group of prime order $n$, and $g$ is a generator of $\mathbb{G}$. The CDH problem is defined as: Given the tuple of elements $ (g,g^x,g^y)$, where $\{x,y\} \leftarrow \mathbb{Z}_n$, there is no efficient algorithm $\mathcal{A}_{CDH}$ that can compute $g^{xy}$.

\subsection{Structure-preserving Constant-size Signature}
\label{CSIG}

Structure-preserving constant-size signature was defined by Abe \emph{et al.} \cite{XCGSIG} as the main scheme of structure-preserving signatures used to sign a message $\Vec{m} = (m_1,...,m_k) \in {\mathbb{G}_2}^k$, considering an asymmetric bilinear group $(n,\mathbb{G}_1,\mathbb{G}_2,$ $\mathbb{G}_3,g_1, g_2, e)$.
A constant-size signature scheme $\mathsf{CSIG}$ \cite{XCGSIG} relies on the following three PPT algorithms ($\mathsf{CSIG}.\mathsf{Key}$, $\mathsf{CSIG}.\mathsf{Sign}$, $\mathsf{CSIG}.\mathsf{Verify}$): \\

${\mathsf{CSIG}}.\mathsf{Key}(1^{\lambda})$: This algorithm takes as input the security parameter $(1^{\lambda})$ and outputs the pair of public and secret keys $(\mathtt{sk}, \mathtt{pk})$ of the signer. It chooses two random generators $g_r, h_u \leftarrow \mathbb{G}^*_1 $ and random values $\gamma_i, \delta_i \leftarrow \mathbb{Z}^*_n$ and computes $g_{i} = {g_r}^{\gamma_i}$ and $h_{i} = {h_u}^{\delta_i}$, for $i= 1,...,k$. It then selects $\gamma_z, \delta_z \leftarrow \mathbb{Z}^*_n$ and computes $g_{z} = {g_r}^{\gamma_z}$ and $h_{z} = {h_u}^{\delta_z}$. It also chooses $\alpha, \beta \leftarrow \mathbb{Z}^*_n$ and sets the couples $(g_r, g_2^{\alpha})$ and $(h_u, g_2^{\beta})$. The public key is set as $\mathtt{pk}= (g_z, h_z, g_r, h_u, g_2^{\alpha}, g_2^{\beta}, \{g_{i}, h_{i}\}^k_{i=1})$ and the secret key is set as $\mathtt{sk} = (\mathtt{pk}, \alpha, \beta, \gamma_z, \delta_z, \{\gamma_i, \delta_i\}^k_{i=1})$.  \\

$\mathsf{CSIG}.\mathsf{Sign}(\mathtt{sk}, \Vec{m})$: This algorithm generates a signature $\sigma$ over a message $\Vec{m}$ using the secret key $\mathtt{sk}$. That is, the signer randomly selects $\zeta, \rho, \tau, \varphi, \omega \leftarrow \mathbb{Z}^*_n$ and computes 
$$z = g_2^{\zeta}, r = {g_2}^{\alpha-\rho \tau - {\gamma_z} \zeta} {\prod}^k_{i=1} {m_i}^{-\gamma_i}, s= {g_r}^{\rho}, t = {g_2}^{\tau},$$
$$u = {g_2}^{\beta-\varphi \omega - {\delta_z} \zeta}{\prod}^k_{i=1} {m_i}^{-\delta_i}, v = {h_u}^{\varphi}, w = {g_2}^{\omega} $$
The signature is set as $\sigma = (z, r, s, t, u, v, w)$. \\

$\mathsf{CSIG}.\mathsf{Verify}(\mathtt{pk}, \Vec{m}, \sigma)$: This algorithm checks the validity of the signature $\sigma$ on the message $m$ relying on the signer's public key $\mathtt{pk}$. It outputs 1 if the signature is valid and 0 otherwise. The verifier checks if the following equations hold: 

\begin{equation}
    A = e(g_z, z)e(g_r, r)e(s,t){\prod}^k_{i=1}e(g_{i}, m_i)
    \label{sigverif1}
\end{equation}

\begin{equation}
    B = e(h_z, z)e(h_u, u)e(v,w){\prod}^k_{i=1}e(h_{i}, m_i)
    \label{sigverif2}
\end{equation}

where $A = e(g_r, g_2^{\alpha})$ and $B= e(h_u, g_2^{\beta})$

\subsection{Structure-preserving signature on mixed-group messages}
\label{XSIG}

A structure-preserving signature on mixed-group messages $\mathsf{XSIG}$ \cite{XCGSIG} represents a signature scheme where the message space is a mixture of the two groups $\mathbb{G}_1$ and $\mathbb{G}_2$. We consider two constant-size signature schemes $\mathsf{CSIG}1$ and $\mathsf{CSIG}2$. $\mathsf{CSIG}2$ is the same scheme as in Section \ref{CSIG} where the message space is ${\mathbb{G}_2}^{k_2}$, while $\mathsf{CSIG}1$ is a ’dual’ scheme obtained by exchanging $\mathbb{G}_1$ and $\mathbb{G}_2$ in the same scheme, where the message space is ${\mathbb{G}_1}^{k_1}$. The message space for the $\mathsf{XSIG}$ is then ${\mathbb{G}_1}^{k_1} \times {\mathbb{G}_2}^{k_2}$. Let $(\Vec{m}, \Vec{\tilde{m}})$ be a message in ${\mathbb{G}_1}^{k_1} \times {\mathbb{G}_2}^{k_2}$. For a vector $\Vec{\tilde{m}} \in {\mathbb{G}_1}^{k_1}$ and a single element $s \in \mathbb{G}_1$, let $\Vec{m}||s$ denote a vector in ${\mathbb{G}_1}^{k_1+1}$ obtained by appending $s$ to the end of $\Vec{m}$.

A mixed-group messages signature scheme $\mathsf{XSIG}$ relies on the following three PPT algorithms ($\mathsf{XSIG}.\mathsf{Key}$, $\mathsf{XSIG}.\mathsf{Sign}$, $\mathsf{XSIG}.\mathsf{Verify}$): \\ 

$\mathsf{XSIG}.\mathsf{Key}(1^{\lambda})$: This algorithm runs $(\mathtt{sk}_1$, $\mathtt{pk}_1)$ $\leftarrow$ $\mathsf{CSIG}1.\mathsf{Key}(1^{\lambda})$ and $(\mathtt{sk}_2$, $\mathtt{pk}_2)$ $\leftarrow$ $\mathsf{CSIG}2.\mathsf{Key}(1^{\lambda})$ and sets $(\mathtt{sk}, \mathtt{pk}) = ((\mathtt{sk}_1, \mathtt{sk}_2),(\mathtt{pk}_1, \mathtt{pk}_2))$. \\

$\mathsf{XSIG}.\mathsf{Sign}(\mathtt{sk}, (\Vec{m}, \Vec{\tilde{m}}))$: This algorithm runs $\sigma_2$ = $(z, r, s, t, u, v, w)$ $\leftarrow$ $\mathsf{CSIG}2.\mathsf{Si}$- $\mathsf{gn}(\mathtt{sk_2}$, $\Vec{\tilde{m}})$ and  $\sigma_1 = (z', r', s', t', u', v', w') \leftarrow \mathsf{CSIG}1.\mathsf{Sign}(\mathtt{sk_1}, \Vec{m}||s)$, and outputs $\sigma = (\sigma_1, \sigma_2)$. \\

$\mathsf{XSIG}.\mathsf{Verify}(\mathtt{pk}, (\Vec{m}, \Vec{\tilde{m}}), (\sigma_1, \sigma_2))$: This algorithm takes $s \in \mathbb{G}_1$ from $\sigma_2$, runs $b_2 = \mathsf{CSIG}2.\mathsf{Verify}(\mathtt{pk_2}, \Vec{\tilde{m}}, \sigma_2)$ and $b_1 = \mathsf{CSIG}1.\mathsf{Verify}(\mathtt{pk_1}, \Vec{m}||s, \sigma_1)$. If $b_1 = b_2 = 1$, the algorithm outputs 1, otherwise it outputs 0.

\subsection{Group signatures drawn from structure-preserving signatures}
\label{sec:gs}

We present hereafter an instantiation of a group signature scheme that allows to sign a group element message relying on a constant-size signature scheme $\mathsf{CSIG}$, a mixed-group messages signature scheme $\mathsf{XSIG}$ and a witness indistinguishable proof of knowledge system $\mathsf{NIWI}$ \cite{grothsahai} (cf. Appendix \ref{NIWI}).

A group signature scheme $\mathsf{GSIG}$ relies on the four following algorithms ($\mathsf{GSIG}.\mathsf{Setup}$, $\mathsf{GSIG}.\mathsf{Join}$, $\mathsf{GSIG}.\mathsf{Sign}$, $\mathsf{GSIG}.\mathsf{Verify}$): 

$\mathsf{GSIG}.\mathsf{Setup}$ : represents the setup algorithm. It runs $\mathsf{XSIG}.\mathsf{Key}$ algorithm that generates the key pair $(\mathtt{sk_g}, \mathtt{pk}_g)$ of the group manager and sets up a CRS $\Sigma_{\mathsf{NIWI}}$ for the $\mathsf{NIWI}$ proof. The group verification key is set as $\mathtt{vk}_g = (\mathtt{pk}_g, \Sigma_{\mathsf{NIWI}})$, while the certification secret key $\mathtt{sk_g}$ is privately stored by the group manager. 

$\mathsf{GSIG}.\mathsf{Join}$: represents the join algorithm. It is composed of two steps. In the first one, the group member generates his key-pair $(\mathtt{sk_p}, \mathtt{pk}_p)$ while running the $\mathsf{CSIG}.\mathsf{Key}$ algorithm. Only the public key $\mathtt{pk}_p$ is sent to the group manager. This latter generates a signature $\sigma_p$ over $\mathtt{pk}_p$, using the $\mathsf{XSIG}.\mathsf{Sign}$ algorithm, and sends it to the group member. 

$\mathsf{GSIG}.\mathsf{Sign}$: represents the signing algorithm run by a group member on a message $m \in \mathbb{G}_2$. The group member generates, over the message $m$, a signature $\sigma_m \leftarrow \mathsf{CSIG}.\mathsf{Sign}(\mathtt{sk}_p, m)$ and a non-interactive witness indistinguishable proof of knowledge $\pi \leftarrow \mathsf{NIWI}.\mathsf{Proof}(\Sigma_{\mathsf{NIWI}}, pub, wit)$ that proves $1 = \mathsf{XSIG}.\mathsf{Verify}(\mathtt{pk}_g, \mathtt{pk}_p, \sigma_p)$ and $1 = \mathsf{CSIG}.\mathsf{Verify}(\mathtt{pk}_p, m, \sigma_m)$ with respect to the witness $wit = (\mathtt{pk}_p, \sigma_p, \sigma_m)$ and the public information $pub = (\mathtt{pk}_g, m)$. The signing algorithm outputs the group signature $\pi$. 

$\mathsf{GSIG}.\mathsf{Verify}$: represents the group signature verification algorithm run by a verifier. It takes $(\mathtt{vk}_g, m, \pi)$ as input and verifies the correctness of the $\mathsf{NIWI}$ proof $\pi$ w.r.t. $pub = (\mathtt{pk}_g, m)$ and the CRS $\Sigma_{\mathsf{NIWI}}$.

\section{SPOT Algorithms}
\label{sec:constr}

This section gives a concrete construction of the different phases and algorithms of \textsf{SPOT}, introduced in Section~\ref{sec:sysmodel}. 
\textsf{SPOT} relies on the different variants of structure-preserving signatures represented in Appendix~\ref{sec:blocks}.

    \subsection{Sys\_Init phase}

    \begin{itemize}
        \item $\mathsf{Set\_params}$ -- a trusted authority sets an asymmetric bilinear group $(n$, $\mathbb{G}_1$, $\mathbb{G}_2$, $\mathbb{G}_3$, $g_1$, $g_2$, $e)$ relying on the security parameter $\lambda$, where $\mathbb{G}_1$ and $\mathbb{G}_2$ are two cyclic groups of prime order $n$, $g_1$ and $g_2$  are generators of respectively $\mathbb{G}_1$ and $\mathbb{G}_2$ and $e$ is a bilinear map such that $e:\mathbb{G}_1 \times \mathbb{G}_2 \rightarrow \mathbb{G}_3$. The trusted authority also considers a cryptographic hash function $\mathbf{H}: \{0,1\}^* \rightarrow \mathbb{Z}_n$. The output of the $\mathsf{Set\_params}$ algorithm represents the system global parameters that are known by all the system entities. The tuple $(n,\mathbb{G}_1,\mathbb{G}_2,\mathbb{G}_3,g_1, g_2, e, \mathbf{H})$ is denoted by $pp$, and is considered as a default input of all algorithms.
        
        \item $\mathsf{HA\_keygen}$ -- a trusted authority takes as input the public parameters $pp$, selects a random $x \in \mathbb{Z}^*_n$ and generates the pair of secret and public keys $(\mathtt{sk}_{\mathcal{HA}}, \mathtt{pk}_{\mathcal{HA}})$ of the health authority as follows:\\
        \centerline{
  $\mathtt{sk}_{\mathcal{HA}}=x \quad ; \quad \mathtt{pk}_{\mathcal{HA}} = g_2^x$}

         \item $\mathsf{S\_keygen}$ -- a trusted authority generates the pair of secret and public keys $(\mathtt{sk}_{\mathcal{S}}, \mathtt{pk}_{\mathcal{S}})$ of the server as given below, relying on the system public parameters $pp$ and two selected randoms $y_1, y_2 \in \mathbb{Z}^*_n$.\\
        \centerline{
  $\mathtt{sk}_{\mathcal{S}}=(y_1, y_2) \quad ; \quad \mathtt{pk}_{\mathcal{S}} = (Y_1, Y_2) = (g_2^{y_1}, g_2^{y_2})$ }

        \item $\mathsf{Setup\_ProxyGr_{\mathcal{GM}}}$ -- $\mathcal{GM}$ sets up the group of proxies by generating a group public key $\mathtt{vk}_{g}$ and a certification secret key $\mathtt{sk}_{g}$ as shown in Algorithm \ref{alg:setup}.

    \item $\mathsf{Join\_ProxyGr_{\mathcal{P}/\mathcal{GM}}}$ -- $\mathcal{P}$ first generates his pair of keys $(\mathtt{sk}_{p}, \mathtt{pk}_{p})$ w.r.t. the ${\mathsf{CSIG}}.\mathsf{Key}$ algorithm (cf. Section \ref{CSIG}). Afterwards, $\mathcal{GM}$ generates a signature $\sigma_p$ over the public key $\mathtt{pk}_{p}$ w.r.t. the ${\mathsf{XSIG}}.\mathsf{Sign}$ algorithm (cf. Section \ref{XSIG}). The $\mathsf{Join\_ProxyGr_{\mathcal{P}/\mathcal{GM}}}$ algorithm is detailed in Algorithm \ref{algo:join}. 

 \def\ALGO{\textsc{ProxyGr.Setup}}
 \begin{algorithm}[!ht]
 \caption{ $\mathsf{Setup\_ProxyGr_{\mathcal{GM}}}$ algorithm}
 \label{alg:setup}
 \begin{algorithmic}[1]
 \State {\bf Input:} the system public parameters $pp$
 \State {\bf Output:} the public parameters $\mathtt{vk}_{g}$ of the proxies' group and the secret key $\mathtt{sk}_{g}$ 

 \State // \textbf{The next iterations are executed to generate the pair of keys of $\mathcal{GM}$}
 \State pick at random $g_{r1}, h_{u1} \leftarrow \mathbb{G}^*_1$, $g_{r2}, h_{u2} \leftarrow \mathbb{G}^*_2$\;
 \For{$i= 1 $ to $2$}{
    pick at random $\gamma_{1i}, \delta_{1i} \leftarrow \mathbb{Z}^*_n$\;
    compute $g_{1i} \leftarrow {g_{r1}}^{\gamma_{1i}}$, $h_{1i} \leftarrow {h_{u1}}^{\delta_{1i}}$\;
 }
 \EndFor
  \For{$j= 1 $ to $7$}{
   pick at random $\gamma_{2j}, \delta_{2j} \leftarrow \mathbb{Z}^*_n$\;
    compute $g_{2i} \leftarrow {g_{r2}}^{\gamma_{2j}}$ and $h_{2j} \leftarrow {h_{u2}}^{\delta_{2j}}$\;
 }
 \EndFor
 \State pick at random $\gamma_{1z}, \delta_{1z}, \gamma_{2z}, \delta_{2z} \leftarrow \mathbb{Z}^*_n$ ;
 \State compute $g_{1z} \leftarrow {g_{r1}}^{\gamma_{1z}}$, $h_{1z} \leftarrow{h_{u1}}^{\delta_{1z}}$, $g_{2z} \leftarrow{g_{r2}}^{\gamma_{2z}}$ and $h_{2z} \leftarrow{h_{u2}}^{\delta_{2z}}$ ;
 \State pick at random $\alpha_1, \alpha_2, \beta_1, \beta_2 \leftarrow \mathbb{Z}^*_n$ ;
 \State $\mathtt{pk}_1 \leftarrow (g_{2z}, h_{2z}, g_{2r}, h_{2u}, g_1^{\alpha_2}, g_1^{\beta_2}, \{g_{2j}, h_{2j}\}^7_{j=1})$ and $\mathtt{sk}_1  \leftarrow (\mathtt{pk}_1, \alpha_2, \beta_2, \gamma_{2z}, \delta_{2z}, \{\gamma_{2j}, \delta_{2j}\}^7_{j=1} )$ ;
 \State $\mathtt{pk}_2 \leftarrow (g_{1z}, h_{1z}, g_{1r}, h_{1u}, g_2^{\alpha_1}, g_2^{\beta_1}, \{g_{1i}, h_{1i}\}^2_{i=1})$ and $\mathtt{sk}_2  \leftarrow (\mathtt{pk}_2, \alpha_1, \beta_1, \gamma_{1z}, \delta_{1z}, \{\gamma_{1i}, \delta_{1i}\}^2_{i=1} )$ ;
  \State set $\mathtt{pk}_g \leftarrow (\mathtt{pk}_1, \mathtt{pk}_2)$ and $\mathtt{sk}_g \leftarrow (\mathtt{sk}_1, \mathtt{sk}_2)$ ;
   \State // \textbf{The next iterations are executed to generate the CRS $\Sigma_{\mathsf{NIWI}}$}
 \State pick at random  $r, s \leftarrow \mathbb{Z}^*_n$ and set $\mathcal{U}= r g_1$ and $\mathcal{V} = s g_2$ ;
 \State set $\Sigma_{\mathsf{NIWI}}= (\mathbb{G}_1,\mathbb{G}_2,\mathbb{G}_3, e, \iota_1, p_1, \iota_2, p_2, \iota_3, \mathcal{U}, \mathcal{V})$ ;
 \State $\mathtt{vk}_g  \leftarrow (\mathtt{pk}_g, \Sigma_{\mathsf{NIWI}})$ ;
\State \textbf{return} $(\mathtt{sk}_g, \mathtt{vk}_g)$
\end{algorithmic}
\end{algorithm}
        
        \item  $\mathsf{Set\_UserID_{\mathcal{HA}}}$ -- every time, a user ($\mathcal{U}$) installs the application and wants to register, $\mathcal{HA}$ picks a secret $t_{\mathcal{U}} \in \mathbb{Z}^*_n$ and sets the user's identifier $\mathtt{ID}_\mathcal{U}$  as\\
        \centerline{
  $\mathtt{ID}_\mathcal{U} = h_{\mathcal{U}} = g_2^{t_{\mathcal{U}}}$}
        
        \item $\mathsf{Userkeygen_{\mathcal{U}}}$ -- After receiving his identifier $\mathtt{ID}_\mathcal{U}$, a user generates his pair of secret and private keys $(\mathtt{sk}_{\mathcal{U}}, \mathtt{pk}_{\mathcal{U}})$. Indeed, $\mathcal{U}$ randomly selects $q_{\mathcal{U}} \in \mathbb{Z}^*_n$ and sets $(\mathtt{sk}_{\mathcal{U}}, \mathtt{pk}_{\mathcal{U}})$ as \\ 
        \centerline{
  $\mathtt{sk}_{\mathcal{U}} = q_{\mathcal{U}} \quad ; \quad \mathtt{pk}_{\mathcal{U}} = {h_{\mathcal{U}}}^{q_{\mathcal{U}}}$ }

    \end{itemize}

   \subsection{Generation phase}
    
    \begin{itemize}
        \item  $\mathsf{Set\_CCM_{\mathcal{U}}}$ -- For each epoch $e$, $\mathcal{U}_A$ and $\mathcal{U}_B$ generate random EBIDs $\mathtt{D}^e_{\mathcal{U}_A}$ and $\mathtt{D}^e_{\mathcal{U}_B}$, respectively. $\mathcal{U}_A$ and $\mathcal{U}_B$ exchange their EBIDs and each of them executes the $\mathsf{Set\_CCM}$ algorithm. $\mathcal{U}_A$ (resp. $\mathcal{U}_B$) computes $m^e_{AB} = \mathtt{D}^e_{\mathcal{U}_A} * \mathtt{D}^e_{\mathcal{U}_B}$ and sets the common contact element between $\mathcal{U}_A$ and $\mathcal{U}_B$ as $\mathtt{CCM}^e_{AB}= \mathbf{H}(m^e_{AB})$.
        
        \item $\mathsf{S\_PSign_{\mathcal{S}}}$ -- After checking that he receives two copies of $\mathtt{CCM}^e_{AB}$, the server picks at random $r_s \leftarrow \mathbb{Z}^*_n$ and, relying on his secret key $\mathtt{sk}_{\mathcal{S}}$, he computes the two messages $\mathtt{PS}^e_{AB}$ and $\mathtt{PS'}^e_{AB}$ such that \\
        \centerline{$\mathtt{PS}^e_{AB} = \mathtt{CCM}^e_{AB} y_1  r_s + y_2 \quad and \quad \mathtt{PS'}^e_{AB} = \mathtt{CCM}^e_{AB}  r_s$}
        
        
        \item $\mathsf{P\_Sign_{\mathcal{P}}}$ -- We consider that when being requested by a user $\mathcal{U}_A$, the proxy opens a session and saves the user's identifier $\mathtt{ID}_{\mathcal{U}_A}$. This latter is used when executing the $\mathsf{P\_Sign_{\mathcal{P}}}$ algorithm (c.f. Algorithm \ref{algo:sign}) to generate a new message $\mathtt{M}^e_{AB}$ (Line 4). The proxy then signs $\mathtt{M}^e_{AB}$ (Line 6 -- Line 8) following the ${\mathsf{CSIG}}.\mathsf{Sign}$ algorithm and finally generates a proof $\pi$ (Line 10 -- Line 16) w.r.t. the ${\mathsf{GSIG}}.\mathsf{Sign}$ algorithm.

\def\ALGO{\textsc{ProxyGr.Join}}
 \begin{algorithm}[!ht]
 \caption{ $\mathsf{Join\_ProxyGr_{\mathcal{P}/\mathcal{GM}}}$ algorithm}
 \label{algo:join}
 \begin{algorithmic}[1]
 \State {\bf Input:} the security parameter $\lambda$ and the secret key of the group manager $\mathtt{sk}_g$
 \State {\bf Output:} the pair of keys of a proxy group member $(\mathtt{sk}_p, \mathtt{pk}_p)$ and the signature $\sigma_p$ over the public key the public $\mathtt{pk}_p$ 
 \medskip
 \State // \textbf{The next is set by $\mathcal{P}$}
 \State pick at random $g_{r}, h_{u} \leftarrow \mathbb{G}^*_1$, $\gamma, \delta \leftarrow \mathbb{Z}^*_n$ ;
 \State compute $g_{\gamma} \leftarrow {g_{r}}^{\gamma}$ and $h_{\delta} \leftarrow {h_{u}}^{\delta}$ ;
 \State pick at random $\gamma_{z}, \delta_{z} \leftarrow \mathbb{Z}^*_n$ ;
 \State compute $g_{z} \leftarrow {g_{r}}^{\gamma_{z}}$ and $h_{z} \leftarrow{h_{u}}^{\delta_{z}}$ ;
 \State pick at random $\alpha, \beta \leftarrow \mathbb{Z}^*_n$ ;
 \State set $\mathtt{pk}_p= (g_{z}, h_{z}, g_{r}, h_{u}, g_2^{\alpha}, g_2^{\beta}, g_{\gamma}, h_{\delta})$ and $\mathtt{sk}_p = (\mathtt{pk}_p, \alpha, \beta, \gamma_{z}, \delta_{z}, \gamma, \delta)$ ;
 \State // \textbf{The next is set by $\mathcal{GM}$}
 \State $\sigma_p \leftarrow \mathsf{XSIG}.\mathsf{Sign}(\mathtt{sk}_g, \mathtt{pk}_p)$ ;
\State \textbf{return} $(\mathtt{sk}_p, \mathtt{pk}_p, \sigma_p)$
\end{algorithmic}
\end{algorithm}

     \def\ALGO{\textsc{M.Sign}}
 \begin{algorithm*}[!ht]
 \caption{ $\mathsf{P\_Sign_{\mathcal{P}}}$ algorithm}
 \label{algo:sign}
 \begin{algorithmic}[1]
 \State {\bf Input:} the public parameters of the proxies' group $\mathtt{vk}_g$, the secret key $\mathtt{sk}_p$, the signature $\sigma_p$ over the proxy's public key, the identifier $\mathtt{ID}_{\mathcal{U}_A}$ of user $\mathcal{U}_A$ and the message $\mathtt{PS}$
 \State {\bf Output:} a message $\mathtt{M}$, the corresponding signature $\sigma_m$ and a proof $\pi$
\medskip
\State // \textbf{The next is executed by $\mathcal{P}$ to generate $\mathtt{M}$}
\State compute $\mathtt{M} = {\mathtt{ID}_{\mathcal{U}_A}}^{\mathtt{PS}}$;
\State // \textbf{The next is executed by $\mathcal{P}$ to sign $\mathtt{M}$}
\State pick at random $\zeta, \rho, \tau, \varphi, \omega \leftarrow \mathbb{Z}^*_n$ ;
\State run $z = g_2^{\zeta}$, $r = {g_2}^{\alpha-\rho \tau - {\gamma_{z}} \zeta}  {\mathtt{M}}^{-\gamma}$, $s= {g_{r}}^{\rho}$, $t = {g_2}^{\tau}$, $u = {g_2}^{\beta-\varphi \omega - {\delta_{z}} \zeta}{\mathtt{M}}^{-\delta}$, $v = {h_{u}}^{\varphi}$, $w = {g_2}^{\omega}$ ;
\State set $\sigma_m = (z, r, s, t, u, v, w)$ ;
\State // \textbf{The next is set to generate a proof on equations $\{(\Vec{\mathcal{A}_{im}}, \Vec{\mathcal{B}_{im}}, \Gamma_{im}, t_{im})\}^2_{i=1}$ where $\Vec{\mathcal{A}_{im}}$ = $\Vec{\mathcal{B}_{im}}$ = $\Vec{0}$, $\Gamma_{im}$ = $\mathcal{MAT}_{3 \times 3}(1)$ for $i=1,2$, $t_{1m}$ = $t_{2m}$ = $1_{\mathbb{G}_3}$}
\State $\Vec{\mathcal{X}}_{1m} = (g_z, g_r, s)$, $\Vec{\mathcal{X}}_{2m} = (h_z, h_u, v) $, $\Vec{\mathcal{Y}}_{1m} = (z, {g_2}^{\alpha-\rho \tau - {\gamma_{z}} \zeta}, t)$ and $\Vec{\mathcal{Y}}_{2m} = (z, {g_2}^{\beta-\rho \tau - {\delta{z}} \zeta}, w)$ ;
\State $ \pi_m = \{(\Vec{\mathcal{C}_{im}}, \Vec{\mathcal{D}_{im}}, \pi_{im}, \theta_{im})\}^2_{i=1} \leftarrow \mathsf{NIWI}.\mathsf{Proof}(\mathtt{vk}_g$, $\{(\Vec{\mathcal{A}_{im}}, \Vec{\mathcal{B}_{im}}, \Gamma_{im}, t_{im})\}^2_{i=1},$ $\{(\Vec{\mathcal{X}_{im}}, \Vec{\mathcal{Y}_{im}})\}^2_{i=1})$ ;
\State // \textbf{The next is set to generate a proof on equations $\{(\Vec{\mathcal{A}_{ip}}, \Vec{\mathcal{B}_{ip}}, \Gamma_{ip}, t_{ip})\}^4_{i=1}$ where $\Vec{\mathcal{A}}_{1p} = (g_1^{\alpha_2})$, $\Vec{\mathcal{A}}_{2p} = (g_1^{\beta_2})$, $\Vec{\mathcal{A}}_{3p} = (g_{1z}, g_{1r}) $, $\Vec{\mathcal{A}}_{4p} = (h_{1z}, h_{1u})$, $\Vec{\mathcal{B}}_{1p} = (g_{2z}, g_{2r}) $, $\Vec{\mathcal{B}}_{2p} = (h_{2z}, h_{2u}) $, $\Vec{\mathcal{B}}_{3p} = (g_2^{\alpha_1})$, $\Vec{\mathcal{B}}_{4p} = (g_2^{\beta_1})$, $\Gamma_{1p} = (\gamma_{2z}, -1)$, $\Gamma_{2p} = (\delta_{2z}, -1)$, $\Gamma_{3p} = (\gamma_{1z}, -1)$, $\Gamma_{4p} = (\delta_{1z}, -1)$, $t_{1p}= e(g_1^{\alpha_2}, g_{2r})$, $t_{2p}= e(g_1^{\beta_2}, h_{2u})$, $t_{3p}= e(g_{1r}, g_2^{\alpha_1})$ and $t_{4p}= e(h_{1u}, g_2^{\beta_1})$ }
\State  $\Vec{\mathcal{X}}_{1p} = (z_1, {g_1}^{\alpha_2-\rho_1 \tau_1 - {\gamma_{2z}} \zeta_1})$, $\Vec{\mathcal{X}}_{2p} = (z_1, {g_1}^{\beta_2-\rho_1 \tau_1 - {\delta_{2z}} \zeta_1})$, $\Vec{\mathcal{X}}_{3p}= (g_{1r}) $, $\Vec{\mathcal{X}}_{4p} = (h_{1u}) $, $\Vec{\mathcal{Y}}_{1p} = (g_{2r}) $,  $\Vec{\mathcal{Y}}_{2p} = (h_{2u})$,    $\Vec{\mathcal{Y}}_{3p} = (z_2, {g_2}^{\alpha_1-\rho_2 \tau_2 - {\gamma_{1z}} \zeta_2})$,  and  $\Vec{\mathcal{Y}}_{4p} = (z_2, {g_2}^{\beta_1-\rho_2 \tau_2 - {\delta_{1z}} \zeta_2})$ ;
\State $ \pi_p = \{(\Vec{\mathcal{C}_{ip}}, \Vec{\mathcal{D}_{ip}}, \pi_{ip}, \theta_{ip})\}^4_{i=1} \leftarrow \mathsf{NIWI}.\mathsf{Proof}(\mathtt{vk}_g, \{(\Vec{\mathcal{A}_{ip}}, \Vec{\mathcal{B}_{ip}}, \Gamma_{ip}, t_{ip})\}^4_{i=1},$ $\{(\Vec{\mathcal{X}_{ip}}, \Vec{\mathcal{Y}_{ip}})\}^4_{i=1})$ ;
\State set $\pi_p= ((\pi_{ip}, \theta_{ip})_{i=1}^4)$ ; 
\State set $\pi = (\pi_p, \pi_m)$ ;
\State \textbf{return} $(\mathtt{M}, \sigma_m, \pi)$
\end{algorithmic}
\end{algorithm*}

     \end{itemize}

    \subsection{Verification phase}
    \begin{itemize}
        \item $\mathsf{Sig\_Verify_{\mathcal{HA}}}$ -- Given a contact list of user $\mathcal{U}_A$ (a list of tuples $(\mathtt{CCM}$, $\mathtt{M}$,  $\pi)$ such that $\pi$ can be parsed as $\{(\Vec{\mathcal{A}_i}, \Vec{\mathcal{B}_i}, \Gamma_i, t_i)\}^N_{i=1}$, $\{(\Vec{\mathcal{C}_i}, \Vec{\mathcal{D}_i}, \pi_i, \theta_i)\}^N_{i=1}$), $\mathcal{HA}$ verifies the validity of the group signature of each message, w.r.t. ${\mathsf{GSIG}}.\mathsf{Verify}$ algorithm (cf. Appendix \ref{sec:gs}).
        
        \item $\mathsf{CCM\_Verify_{\mathcal{HA}}}$ -- We consider that $\mathcal{HA}$ requests from $\mathcal{S}$ the message $\mathtt{PS'}$ corresponding to a contact message $\mathtt{CCM}$ contained in the contact list of user $\mathcal{U}_A$. The message $\mathtt{PS'}$ is taken as input with the message $\mathtt{M}$ (corresponding to $\mathtt{CCM}$), the server's public key $\mathtt{pk}_{\mathcal{S}}$ and the secret value $t_{\mathcal{U}_A}$ specific to user $\mathcal{U}_A$, to the $\mathsf{CCM\_Verify_{HA}}$ algorithm that checks if the equation \ref{newequverify} holds: 
        
            
             \begin{equation}
                \mathtt{M} = {Y_1}^{t_{\mathcal{U}_A} \mathtt{PS'}} {Y_2}^{t_{\mathcal{U}_A}}
                 \label{newequverify}
            \end{equation}
\end{itemize}

\section{Security and Privacy Analysis}
\label{sec:secanal}

In this section, we prove that \textsf{SPOT} achieves the defined security and privacy requirements with respect to the threat models defined in Section \ref{sec:thrmodel}, by relying on the following theorems and lemmas. 


\begin{theorem}[Unforgeability]
\label{theo:unforg2}
If a probabilistic-polynomial time (PPT) adversary $\mathcal{A}$ wins $\textbf{Exp}_{\mathcal{A}}^{unforg}$, as defined in Section \ref{sec:unforg}, with a non-negligible advantage $\epsilon$, then a PPT simulator $\mathcal{B}$ can be constructed to break the CDH assumption with a non-negligible advantage $\epsilon$.
\end{theorem}

\begin{proof}
In this proof, we show that a simulator $\mathcal{B}$ can be constructed with the help of an adversary $\mathcal{A}$ having advantage $\epsilon$ against \textsf{SPOT} scheme.

The CDH challenger $\mathcal{C}$ sends to $\mathcal{B}$ the tuple $(g_2,g_2^a,g_2^b)$, where $a,b \leftarrow \mathbb{Z}^*_n$ are randomly selected. $\mathcal{C}$ asks $\mathcal{B}$ to compute $g_2^{ab}$. Then, $\mathcal{B}$ sets $g_2^{t_{\mathcal{U}}}$ to $g_2^a$ and $g_2^{y_2}$ to $g_2^b$. During the challenge phase, $\mathcal{B}$ randomly selects $y_1 \in \mathbb{Z}^*_n$ and sends $g_2^{y_1}$ to $\mathcal{A}$ as part of the server's public key. 
$\mathcal{A}$ forges the partial signature over the message $\mathtt{PS'}$ and generates the message $\mathtt{M}^*$ with advantage $\epsilon$: $\mathtt{M}^*$ = ${\mathtt{ID}_{\mathcal{U}}}^{\mathtt{PS}^*}$ = ${g_2}^{t_{\mathcal{U}} (\mathtt{PS'} y_1 + y_2)}$. The tuple $({g_2}^{t_{\mathcal{U}}},\mathtt{PS'}, \mathtt{M}^*)$ is sent back to $\mathcal{B}$. Upon receiving this tuple and knowing $y_1$, $\mathcal{B}$ can compute the value of $g_2^{t_{\mathcal{U}} y_2}$ which is the same as $g_2^{ab}$ and can then send the result to the CDH challenger. As such, $\mathcal{B}$ succeeds the forgery against the CDH assumption with advantage $\epsilon$.
\end{proof}


\begin{theorem}[Unlinkability]
\label{theo:unlink}
Our \textsf{SPOT} system achieves the unlinkability requirement with respect to the \emph{group-signature unlinkability} and \emph{multi-CCM unlinkability} properties.
\end{theorem}

We prove Theorem \ref{theo:unlink} through Lemma \ref{lem:groupunlink}  and Lemma \ref{lem:multiccmunlink} with respect to \emph{group-signature unlinkability} and \emph{multi-CCM unlinkability} properties, respectively.

\begin{lemma}[Group-signature unlinkability]
\label{lem:groupunlink}
\textsf{SPOT} satisfies the group signature unlinkability requirement with respect to the computational witness indistinguishability property of the NIWI proof. 
\label{gsunlink}
\end{lemma}

\begin{proof}
In this proof, the objective is to show that the adversary is not able to distinguish group signatures issued by the same proxy. 
For this purpose, we suppose that, for each session $i$, the adversary receives the message $M^*$ (i.e., the same message $M^*$ is returned by each oracle) and the NIWI proof $\pi^i$ = $(\pi_m^i, \pi_p^i)$ = $((\pi^i_{jm}, \theta^i_{jm})_{j=1}^2$, $(\pi^i_{jp}, \theta^i_{jp})_{j=1}^4)$. \\
To simplify the proof, we will only consider the NIWI proof $\pi_m^i$, as the statements used to generate the proofs $\pi_p^i$ do not give any information about the proxy generating the proof (i.e., statements do not include the proxy's public key). Thus, for each session $i$, the adversary is given the tuples $(\Vec{\mathcal{C}^i_{k1}}, \Vec{\mathcal{D}^i_{k1}}, \pi^i_{k1}, \theta^i_{k1})$ and $(\Vec{\mathcal{C}^i_{k2}}, \Vec{\mathcal{D}^i_{k2}}, \pi^i_{k2}, \theta^i_{k2})$ referred to as the group signature generated by a proxy $P_k$, where $k \in \{0,1\}$.
During the challenge phase, the adversary is also given two group signatures. The first signature is represented by the tuples $(\Vec{\mathcal{C}^*_{1}}, \Vec{\mathcal{D}^*_{1}}, \pi^*_{1}, \theta^*_{1})$ and $(\Vec{\mathcal{C}^*_{2}}, \Vec{\mathcal{D}^*_{2}}, \pi^*_{2}, \theta^*_{2})$ generated by proxy $\mathcal{P}_0$, while the second one is represented by the tuples $(\Vec{\mathcal{C}^*_{b1}}, \Vec{\mathcal{D}^*_{b1}}, \pi^*_{b1}, \theta^*_{b1})$ and $(\Vec{\mathcal{C}^*_{b2}}, \Vec{\mathcal{D}^*_{b2}}, \pi^*_{b2}, \theta^*_{b2})$ and is generated by a proxy $\mathcal{P}_b$ ($b \in \{0,1\}).$

Let us consider a simulator $\mathcal{B}$ that can be constructed with the help of an adversary $\mathcal{A}$ having advantage $\epsilon$ against \textsf{SPOT} scheme.
A challenger $\mathcal{C}$ selects two couples of witnesses $(X_0, Y_0)$ and $(X_1, Y_1)$. $\mathcal{C}$ computes a commitment $(C, D)$ over $(X_0, Y_0)$, and then selects a bit $b \in \{0,1\}$ and computes a commitment $(C'_b, D'_b)$ over $(X_b, Y_b)$. $\mathcal{C}$ asks $\mathcal{B}$ to guess the bit $b$. Then, $\mathcal{B}$ selects the tuples $(A, B, \Gamma, t)$ and $(A'_b, B'_b, \Gamma'_b, t'_b)$ and computes the proofs $(\pi, \theta)$ and $(\pi'_b, \theta'_b)$. $\mathcal{B}$ returns the two proofs to $\mathcal{A}$. Finally, $\mathcal{A}$ outputs a bit $b'$ that it sends to $\mathcal{B}$. This latter outputs the same bit $b'$ to its own challenger $\mathcal{C}$. As such, $\mathcal{A}$ succeeds in breaking the group-signature unlinkability  with advantage $\epsilon$, which is the same as breaking the computational witness-indistinguishability property.

\end{proof}

\begin{corollary}
If \textsf{SPOT} satisfies the unlinkability property, then the proxies' group signature (i.e., NIWI proof) fulfills the anonymity requirement stating that it is not possible to identify the proxy that issued a particular group signature.  
\end{corollary}

\begin{lemma}[Multi-CCM unlinkability]
\label{lem:multiccmunlink}
\textsf{SPOT} satisfies the multi-CCM unlinkability requirement with respect to the common contact message structure.
\label{mccmunlink}
\end{lemma}

\begin{sproof}
Let $\mathcal{A}$ be a successful adversary against the \emph{multi-CCM unlinkability} property.
Assume that $\mathcal{A}$ receives two messages $\mathtt{CCM_1} = \mathbf{H}(\mathtt{D}_i * \mathtt{D}_j)$ and  $\mathtt{CCM_2} = \mathbf{H}(\mathtt{D}_i * \mathtt{D}_k)$ (with $j \neq k$) meaning that user $U_i$ met $U_j$ and $U_k$, then $\mathcal{A}$ is not able to link $\mathtt{CCM_1}$ and $\mathtt{CCM_2}$ to the same user $\mathcal{U}_i$ as all EBIDs are randomly generated in each epoch $e$, and the hashing function $\mathbf{H}$ behaves as a pseudo-random function.
\end{sproof}


\begin{theorem}[Anonymity]
\label{theo:anon}
\textsf{SPOT} satisfies the anonymity property, in the sense of Definition \ref{def:anon}, if and only if, the CCM-unlinkability requirement is fulfilled.
\end{theorem}

\begin{sproof}
We prove that our proximity-based protocol \textsf{SPOT} satisfies the anonymity property using an \emph{absurdum} reasoning.
We suppose that an adversary $\mathcal{A}$ can break the anonymity of \textsf{SPOT}, in the sense of Definition \ref{def:anon}, by reaching the advantage $Pr[\textbf{Exp}_{\mathcal{A}}^{anon}(1^{\lambda})=1] \geq \frac{1}{2} \pm \kappa(\lambda)$. $\mathcal{A}$ is given the pair of public-private keys $(\mathtt{pk}_{\mathcal{HA}}, \mathtt{sk}_{\mathcal{HA}})$ of the health authority, the identifiers $\mathtt{ID}_\mathcal{U}$  of all users and a contact list of a particular user $\mathcal{U}_A$, obtained when relying on several sessions. Then, relying on the \emph{left-or-right} $\mathsf{LoRCU}$ oracle, $\mathcal{A}$ tries to distinguish the user $\mathcal{U}_b$ being in contact with $\mathcal{U}_A$, better than a flipping coin. That is, given the tuple $({\mathsf{CCM}}^*_{b}, \mathtt{M}^*_b, \pi^*_b)$, $\mathcal{A}$ successfully predicts the identifier $\mathtt{ID}_{\mathcal{U}_b}$. 
Obviously, $\mathcal{A}$ tries to identify user $\mathcal{U}_b$ relying on the message ${\mathsf{CCM}}^*_{b}$, since both $\mathtt{M}^*_b$ and $\pi^*_b$ are generated based on ${\mathsf{CCM}}^*_{b}$ and give no further information about the user $\mathcal{U}_b$. This refers to link the message ${\mathsf{CCM}}^*_{b}$ to its issuer $\mathcal{U}_b$. Thus, if $\mathcal{A}$ succeeds, this means that $\mathcal{A}$ is able to link two or several common contact messages to the same user, which contradicts the \emph{multi-CCM unlinkability} property previously discussed. As such, we prove that the adversary succeeds $\textbf{Exp}_{\mathcal{A}}^{anon}(1^{\lambda})$ with a probability $Pr = \frac{1}{2} \pm \kappa(\lambda)$, where $\kappa(\lambda)$ is negligible. Thus, \textsf{SPOT} satisfies anonymity.

\end{sproof}


\begin{theorem}[Anti-replay]
\label{theo:arep}
\textsf{SPOT} satisfies the anti-replay requirement and supports false positive hindrance, if the proposed scheme is unforgeable. 
\end{theorem}

\begin{sproof}
To successfully replay a common contact message generated in an epoch $e$, in another epoch $e' \neq e$, a malicious user can perform in two ways. (i) The user reinserts, in his contact list, the tuple $(\mathtt{CCM}^e$, $\mathtt{M}^e$,  $\pi^e)$ generated in an epoch $e$. The reinsertion is performed in an epoch $e' > e + \Delta$\footnote{It makes no sense to reinsert an element in an epoch $e' < e + \Delta$, as duplicated messages will be deleted either by the server or at the user's end-device.}. Afterwards, the contact list is sent to $\mathcal{HA}$ when the user is infected. $\mathcal{HA}$ asks the server to provide the message $\mathtt{PS}'$ corresponding to $\mathtt{CCM}^e$. As the server has no entry corresponding to $\mathtt{CCM}^e$ in the last $\Delta$ days, the second verification performed by $\mathcal{HA}$ does not hold and the tuple is rejected. (ii) We assume that, in an epoch $e'> e + \Delta$, the user is able to replay a message $\mathtt{CCM}^e$ with two different proxies and he successfully receives the corresponding message $\mathtt{M}$ and the group signature $\pi$. Thus, when the user is infected, the health authority validates false positives, but this has no impact on the computation of the risk score, as no user has the same entry in his contact list. 
As such, we can prove the resistance of \textsf{SPOT} against replay attacks.
\end{sproof}

\section{Performance Analysis}
\label{sec:perf}

This section introduces \textsf{SPOT} test-bed, discusses the experimental results, presented in Table \ref{tab:perf}, and demonstrates the usability of the proposed construction for real world scenarios.

\subsection{\textsf{SPOT} Test-bed}

For our experiments, we developed a prototype of the \textsf{SPOT} protocol that implements the three phases \textsc{Sys\_Init}, \textsc{Generation} and \textsc{Verification} including the twelve algorithms\footnote{The source code is available at {\url{https://github.com/privteam/SPOT}}}. The tests are made on an Ubuntu $18.04.3$ machine - with an \emph{Intel} Core $i7 @1.30 GHz$ - 4 cores processor and $8 GB$ memory. The twelve algorithms were implemented based on JAVA version $11$, and the cryptographic library \textsl{JPBC}\footnote{\url{http://gas.dia.unisa.it/projects/jpbc/}}.
We evaluate the computation time of each algorithm relying on two types of bilinear pairings, i.e., \textsl{type A} and \textsl{type F}. The pairing \textsl{type A} is the fastest symmetric pairing type in the \textsl{JPBC} library constructed on the curve $y^2 = x^3 +x$ with an embedding degree equal to 2. The pairing \textsl{type F} is an asymmetric pairing type introduced by Barreto and Naehrig \cite{pairingJPBC}. It has an embedding degree equal to 12. For the two types of pairing, we consider two different levels of security i.e., 112-bits and 128-bits security levels recommended by the US National Institute of Standards and Technology\footnote{\url{http://keylength.com}} (NIST). 
\\
Based on the selected cryptographic library and the implementation of Groth-Sahai proofs\footnote{\url{https://github.com/gijsvl/groth-sahai}}, the \textsf{SPOT} test-bed is built with six main java classes, w.r.t. to the different entities of \textsf{SPOT}, referred to as \emph{TrustedAuthority.java}, \emph{GroupManager.java}, \emph{Proxy.java}, \emph{HealthAuthority.java}, \emph{User.java} and \emph{Server.java}. Each class encompasses the algorithms that are performed by the relevant entity as described in Section \ref{sec:sysmodel}.
In order to obtain accurate measurements of the computation time, each algorithm is run 100 times. Thus, the computation times represent the mean of the 100 runs while considering a standard deviation of an order $10^{-2}$.

\subsection{Communication and Computation Performances of SPOT}
\label{sec:ComCompPerf}

This section first proposes a theoretical analysis of the communication cost. Then, it presents the experimental results of the implementation of \textsf{SPOT} algorithms.

\begin{table*}[!ht]
\begin{center}
\caption{Computation time and communication overhead of SPOT algorithms}
\begin{threeparttable}

\resizebox{\textwidth}{!}{
\begin{tabular}{c|c|c|c|c|c|c|c}
\hline
     \multirow{2}{*}{Algorithm} & \multirow{2}{*}{Entity} &  \multirow{2}{*}{Synch/Asynch} & \multirow{2}{*}{Communication cost} & \multicolumn{4}{c}{Computation time (ms)}\\ \cline{5-8}
     &  &  &  & A/112-bits  &  A/128-bits &  F/112-bits  &  F/128-bits \\ \hline 
      
     $\mathsf{Set\_params}$  & $\mathcal{TA}$  & Asynch. & $|\mathbb{Z}_n| + |\mathbb{G}_1| + |\mathbb{G}_2| + |\mathbb{G}_3|$ & 874 & 2521 & 1230 & 1364 \\ \hline
     $\mathsf{HA\_Keygen}$  & $\mathcal{TA}$  & Asynch. & $|\mathbb{G}_1|$ & 59 & 123 & 12 & 16 \\ \hline
     $\mathsf{S\_Keygen}$  &  $\mathcal{TA}$ & Asynch. & $2|\mathbb{G}_2|$ & 119 & 244 & 24 & 31 \\ \hline
     $\mathsf{Setup\_ProxyGr}$  & $\mathcal{GM}$  & Asynch. & $21(|\mathbb{G}_1| + |\mathbb{G}_2|)$ & 1955 & 4075 & 346 & 451\\ \hline
     $\mathsf{Join\_ProxyGr}$  & $\mathcal{P}$/$\mathcal{GM}$  & Synch. & $\mathcal{P}: 8|\mathbb{G}_1| + 2|\mathbb{G}_2|$ / $\mathcal{GM}: 7(|\mathbb{G}_1| + |\mathbb{G}_2|)$ & 2861 & 6014 & 1159 & 1409\\ \hline
     $\mathsf{Set\_UserID}$  & $\mathcal{HA}$  & Synch. & $|\mathbb{G}_2|$ & 58 & 121 & 12 & 16 \\ \hline
     $\mathsf{Userkeygen}$  & $\mathcal{U}$  & Asynch. & $|\mathbb{G}_2|$ & 117 & 242 & 24 & 31 \\  \hline
     $\mathsf{Set\_CCM}$  &  $\mathcal{U}$ & Synch. & $|\mathbb{Z}_n|$ & 0.1 & 0.1 & 0.1 & 0.1 \\ \hline
     $\mathsf{S\_PSign}$ \tnote{a}  & $\mathcal{S}$  & Synch. & $|\mathbb{Z}_n|$ & 0.1 & 0.08 & 0.02 & 0.02 \\ \hline
     $\mathsf{P\_Sign}$ \tnote{a}  & $\mathcal{P}$  & Synch. & $6|\mathbb{G}_1| + 7|\mathbb{G}_2|$ & 19353 & 40371 & 3164 & 4170 \\ \hline
     $\mathsf{Sig\_Verify}$ \tnote{a}  &  $\mathcal{HA}$ & Asynch. & N.A.  & 6541 & 15406 & 31637 & 36892 \\ \hline
     $\mathsf{CCM\_Verify}$ \tnote{a}  & $\mathcal{HA}$  & Asynch. & N.A. & 174 & 360 & 148 & 190 \\ \hline

  \end{tabular}}
\end{threeparttable}
 \label{tab:perf}
\end{center}
\scriptsize NOTE: Synch./Asynch. indicates whether the algorithm must be run online (i.e., in real time) or offline (i.e., later); $^a$ indicates that the algorithm is performed on a single contact message that is generated by the $\mathsf{Set\_CCM}$ algorithm; $|\mathbb{G}_1|$ (resp. $|\mathbb{G}_2|$, $|\mathbb{G}_3|$ and $|\mathbb{Z}_n|$) indicates the size of an element in ${\mathbb{G}_1}$ (resp. ${\mathbb{G}_2}$, ${\mathbb{G}_3}$ and ${\mathbb{Z}_n}$); N.A. is the abbreviation for Not Applicable. 
\end{table*}

As presented in Table \ref{tab:perf}, the communication cost is measured according to the size of the elements ${\mathbb{G}_1}$, ${\mathbb{G}_2}$, ${\mathbb{G}_3}$ and ${\mathbb{Z}_n}$ exchanged between entities. $\mathsf{Setup\_ProxyGr}$ and $\mathsf{Join\_ProxyGr}$ are the most bandwidth consuming algorithms, however this result must be put into perspective as both algorithms are performed once. Other algorithms have acceptable communication overhead, in particular those performed repeatedly by the user, which proves the efficiency of \textsf{SPOT}.

From Table \ref{tab:perf}, it is worth noticing that the computation time depends on the selected pairings types and is strongly related to the security level. 
Some algorithms of the \textsc{Sys\_Init} phase are consuming but they are limited to only one execution from a powerful trusted authority. 
For the \textsc{Generation} phase, the most consuming algorithm is $\mathsf{P\_Sign}$ which requires 19 seconds (resp. 40 seconds) for pairing \textsl{type A} and 3 seconds (resp. 4 seconds) for pairing \textsl{type F}. The computation time of $\mathsf{Set\_CCM}$ and $\mathsf{S\_PSign}$ are negligible, which means that the user and the server are not required to have important computation capacities. Finally, to verify the correctness of a single contact message, the \textsc{Verification} phase requires approximately 7 seconds (resp. 15 seconds) for pairing \textsl{type A} and 32 seconds (resp. 37 seconds) for pairing \textsl{type F}. 

It is clear that $\mathsf{P\_Sign}$ and $\mathsf{Sig\_Verify}$  are the most consuming algorithms in terms of computation time as they include a large number of exponentiations and pairing functions, however this result must be put into perspective as both the proxy and the health authority are assumed to have advanced hardware features. Table \ref{tab:perf} also shows that the $\mathsf{Sig\_Verify}$ algorithm run with pairing \textsl{type A} is faster than with pairing \textsl{type F}, as this latter requires excessive memory allocation and deallocation. The $\mathsf{P\_Sign}$ algorithm has an opposite behavior where the execution with pairing \textsl{type F} is faster than pairing \textsl{type A}, which is compliant to the \textsl{JPBC} library benchmark\footnote{\url{http://gas.dia.unisa.it/projects/jpbc/benchmark.html}} showing that elementary functions of multiplication and exponentiation require less computation time for pairing \textsl{type F}.

From Table \ref{tab:perf}, we can deduce that the algorithms executed at the user's side, have very low computation and communication overhead, which confirms the usability of \textsf{SPOT}, even when being run on a smartphone with low capacities. For both consuming algorithms ($\mathsf{P\_Sign}$ and $\mathsf{Sig\_Verify}$) that are repeatedly run, some performance improvement means are proposed in the next subsection. 

\subsection{Improved Performances with Multithreading and Preprocessing}

For both computation consuming algorithms $\mathsf{P\_Sign}$ and $\mathsf{Sig\_Verify}$ (see Section \ref{sec:ComCompPerf}), in an effort to make the computation time as efficient as possible, although they run on powerful \textsf{SPOT} entities, we rely on a two step improvement:  
\begin{itemize}
    \item \textbf{Multithreading:} applied to both $\mathsf{P\_Sign}$ and $\mathsf{Sig\_Verify}$ algorithms. It enables simultaneous multiple threads execution (e.g., the computation of the different parts of the NIWI proof for the $\mathsf{P\_Sign}$ algorithm, the computation of either the different verification equations of the NIWI proof, or the two sides of each equation, for the $\mathsf{Sig\_Verify}$ algorithm).
    \item \textbf{Preprocessing:} applied only to $\mathsf{Sig\_Verify}$ algorithm. It enables to prepare in advance a value to be later paired several times, like the variables $\mathcal{U}$ and $\mathcal{V}$ which are used as input to pairing functions for each verification equation.  
\end{itemize}

\begin{figure*}
     \centering
     \begin{subfigure}[b]{0.443\textwidth}
         \centering
         \includegraphics[width=\textwidth]{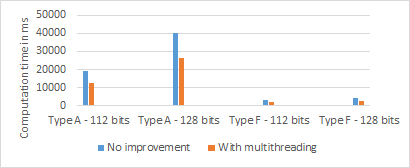}
         \caption{Influence of multithreading on $\mathsf{P\_Sign}$ algorithm}
         \label{fig:psign}
     \end{subfigure}
     \hfill
     \begin{subfigure}[b]{0.548\textwidth}
         \centering
         \includegraphics[width=\textwidth]{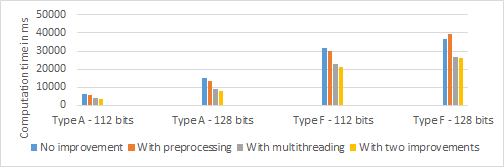}
         \caption{Influence of preprocessing or/and multithreading on $\mathsf{Sig\_Verify}$ algorithm}
         \label{fig:sigverif}
     \end{subfigure}
        \caption{Influence of improvements on $\mathsf{P\_Sign}$ and $\mathsf{Sig\_Verify}$ algorithms}
        \label{fig:improv}
\end{figure*}

Figure \ref{fig:improv} exposes the impacts of one or two combined improvements applied to $\mathsf{P\_Sign}$ and  $\mathsf{Sig\_Verify}$.
From Figure \ref{fig:psign}, we notice that multithreading reduces the computation time for $\mathsf{P\_Sign}$ of approximately $35\%$, for the two types of pairing and the two levels of security. 
For $\mathsf{Sig\_Verify}$, Figure \ref{fig:sigverif} shows that multithreading has a greater impact on the computation time (i.e., approximately $40\%$ for pairing type A and $28\%$ for pairing type F) than preprocessing (i.e., approximately $10\%$ for pairing type A and $5\%$ for pairing type F~\footnote{For type F - 128 bits, the preprocessing decreases the performances. This is due to the excessive memory allocation  and deallocation required  by the pairing type F.}). The two combined improvements ensure a gain of almost $50\%$ for pairing type A and $30\%$ for pairing type F.

\section{Conclusion}
\label{sec:conclusion}

In this paper, a novel secure and privacy-preserving proximity-based \textsf{SPOT} protocol for e-healthcare systems is introduced. The objective of \textsf{SPOT} is to  help governments and healthcare systems to deal with pandemics by automating the process of contact tracing, with security guarantees against fake contacts injection and privacy preservation for users. Thanks to the underlying network architecture relying on a centralized computing server and decentralized proxies, \textsf{SPOT} enables users to determine whether they were in close proximity with infected people, with no risk of false positive alerts. 
The strength of the paper is to provide a full concrete construction of \textsf{SPOT} which is proven to be secure and to support several privacy properties under standard assumptions.
Another strength of the contribution is a PoC of \textsf{SPOT} including a full implementation of the different algorithms, where practical computation costs measurements demonstrate the feasibility of our proposed protocol.\\
Further research will consider aggregating the verification of multiple contact messages in an effort to improve verification performances.



\bibliographystyle{plain}
\bibliography{bibliography}

\appendix

\section{Non-Interactive Witness Indistinguishable Proof }
\label{NIWI}

In this section we represent the Groth-Sahai NIWI proof scheme applied on pairing product equations with an asymmetric bilinear map. Witness-indistinguishability implies that the verifier of a group signature is not able to find the group member that has generated the signature. The $\mathsf{NIWI}$ scheme we consider, involves four PPT algorithms ($\mathsf{NIWI}.\mathsf{Setup}$, $\mathsf{NIWI}.\mathsf{CRS}$, $\mathsf{NIWI}.\mathsf{Proof}$, $\mathsf{NIWI}.\mathsf{Verify}$):

$\mathsf{NIWI}.\mathsf{Setup}$: This algorithm outputs a setup $(\mathtt{gk}, \mathtt{sk})$ such that $\mathtt{gk} = (n$, $\mathbb{G}_1$, $\mathbb{G}_2$, $\mathbb{G}_3$, $g_1$, $g_2$, $e)$ and $\mathtt{sk} = (p, q)$ where $n = p q$.\\

$\mathsf{NIWI}.\mathsf{CRS}$: This algorithm generates a common reference string $\mathtt{CRS}$. It takes $(\mathtt{gk}, \mathtt{sk})$ as inputs and produces $\mathtt{CRS}$ = $(\mathbb{G}_1,\mathbb{G}_2,\mathbb{G}_3, e, \iota_1, p_1, \iota_2, p_2, \iota_3, p_3, \mathcal{U}, \mathcal{V})$, where $\mathcal{U}= r g_1$, $\mathcal{V} = s g_2$ ; $r, s \in \mathbb{Z}^*_n$ and

\begin{center}

\begin{tabular}{c c c c c c}

     $\iota_1$: & $\mathbb{G}_1 \longrightarrow \mathbb{G}_1$ & $\iota_2$: & $\mathbb{G}_2 \longrightarrow \mathbb{G}_2$    & $\iota_3$: & $\mathbb{G}_3 \longrightarrow \mathbb{G}_3$  \\
       & $x \longmapsto x$ &   & $y \longmapsto y$  &  & $z \longmapsto z$  \\ 
       &   &  &  &  &   \\  
    $p_1$: & $\mathbb{G}_1 \longrightarrow \mathbb{G}_1$  & $p_2$: & $\mathbb{G}_2 \longrightarrow \mathbb{G}_2$   & $p_3$: & $\mathbb{G}_3 \longrightarrow \mathbb{G}_3$  \\
       & $x \longmapsto \lambda x$ &  & $y \longmapsto \lambda y$ &   & $z \longmapsto z^{\lambda}$  \\

  \end{tabular}
 \label{tab:crs}
\end{center}

$\mathsf{NIWI}.\mathsf{Proof}$: This algorithm generates a NIWI proof for satisfiability of a set of pairing product equations of the form of 
$${\prod}^l_{i=1}e(\mathcal{A}_{i}, \mathcal{Y}_i) {\prod}^k_{i=1}e(\mathcal{X}_{i}, \mathcal{B}_i) {\prod}^k_{i=1} {\prod}^l_{j=1} e(\mathcal{X}_{i}, \mathcal{Y}_j)^{\gamma_{ij}}=t$$ also written as 
$$(\Vec{\mathcal{A}} \cdot \Vec{\mathcal{Y}})(\Vec{\mathcal{X}} \cdot \Vec{\mathcal{B}})(\Vec{\mathcal{X}} \cdot \Gamma \Vec{\mathcal{Y}})=t$$
It takes as input $\mathtt{gk}$, $\mathtt{CRS}$ and a list of pairing product equations $\{(\Vec{\mathcal{A}_i}, \Vec{\mathcal{B}_i}, \Gamma_i, t_i)\}^N_{i=1}$ and a satisfying witness $\Vec{\mathcal{X}} \in \mathbb{G}_1^k$, $\Vec{\mathcal{Y}} \in \mathbb{G}_2^l$. To generate a proof over a pairing product equation, the algorithm, first, picks at random $\mathcal{R} \leftarrow Vec_k(\mathbb{Z}_n)$ and $\mathcal{S} \leftarrow Vec_l(\mathbb{Z}_n)$, commits to all variables as $\Vec{\mathcal{C}}:=\Vec{\mathcal{X}}+ \mathcal{R} \mathcal{U}$ and $\Vec{\mathcal{D}}:=\Vec{\mathcal{Y}}+ \mathcal{S} \mathcal{V}$, and computes 
$${\pi}= \mathcal{R}^{\top} \iota_2(\Vec{\mathcal{B}}) + \mathcal{R}^{\top} \Gamma \iota_2(\Vec{\mathcal{Y}}) + \mathcal{R}^{\top} \Gamma \mathcal{S} \mathcal{V}$$
$${\theta}= \mathcal{S}^{\top} \iota_1(\Vec{\mathcal{A}}) + \mathcal{S}^{\top} \Gamma^{\top} \iota_1(\Vec{\mathcal{X}})$$
The algorithm outputs the proof $(\pi, \theta)$. \\

$\mathsf{NIWI}.\mathsf{Verify}$: This algorithm checks if the proof is valid. It takes $\mathtt{gk}$, $\mathtt{CRS}$, $\{(\Vec{\mathcal{A}_i}, \Vec{\mathcal{B}_i}, \Gamma_i, t_i)\}^N_{i=1}$ and $(\Vec{\mathcal{C}_i}, \Vec{\mathcal{D}_i}, \{( \pi_i, \theta_i)\}^N_{i=1})$ as inputs and for each equation, checks the following equation:

\begin{equation}
    e(\iota_1(\Vec{\mathcal{A}_i}), \Vec{\mathcal{D}_i}) e(\Vec{\mathcal{C}_i}, \iota_2(\Vec{\mathcal{B}_i}))  e(\Vec{\mathcal{C}_i}, \Gamma_i \Vec{\mathcal{D}_i})= \iota_3(t_i) e(\mathcal{U}, \pi_i) e(\theta_i, \mathcal{V})
    \label{proofverif}
\end{equation}
The algorithm outputs 1 if the equation holds, else it outputs 0.

\end{document}